\documentclass[a4paper,10pt]{article}
\usepackage{a4wide}
\usepackage{natbib}
\usepackage{float}       
\usepackage{amsfonts}
\usepackage{amssymb}
\usepackage{amsmath}
\usepackage{amsfonts}
\usepackage{mathrsfs}
\usepackage{amssymb}
\usepackage[latin1]{inputenc}
\usepackage{lmodern}
\usepackage{natbib}
\usepackage[T1]{fontenc}
\usepackage{graphicx} 
\usepackage{amsmath}  
\usepackage{float}       
\usepackage{enumerate}
\usepackage{bm}
\usepackage{amsthm}
\usepackage{color}

\newcommand{\wto}{\xrightarrow{w}}
\newcommand{\dto}{\rightsquigarrow}

\newtheorem{theorem}{Theorem}

\begin{document}
\author{Miguel de Carvalho\footnote{Ecole Polytechnique F\'ed\'erale de Lausanne,
  Switzerland (Miguel.Carvalho@epfl.ch).}~\footnote{Universidade Nova de Lisboa, Centro
de Matem\'atica e Aplica\c c\~oes, Portugal}, Boris Oumow$^{*}$,
Johan Segers\footnote{Universit\'e catholique de Louvain, Belgium},
and Micha\l~Warcho\l\footnote{Uniwersytet Jagiello\'nski, Poland}
}
\date{\today}
\title{A Euclidean Likelihood Estimator for Bivariate Tail Dependence}
\maketitle

\begin{abstract}
The spectral measure plays a key role in the statistical modeling of multivariate extremes. Estimation of the spectral measure is a complex issue, given the need to obey a certain moment condition. We propose a Euclidean likelihood-based estimator for the spectral measure which is simple and explicitly defined, with its expression being free of Lagrange multipliers. Our estimator is shown to have the same limit distribution as the maximum empirical likelihood estimator of J.~H.~J. Einmahl and J. Segers, Annals of Statistics 37(5B), 2953--2989 (2009). Numerical experiments suggest an overall good performance and identical behavior to the maximum empirical likelihood estimator. We illustrate the method in an extreme temperature data analysis.
\end{abstract}


\section{Introduction}
When modeling dependence for bivariate extremes, only an infinite-dimensional object is flexible enough to capture the `spectrum' of all possible types of dependence. One of such infinite-dimensional objects is the spectral measure, describing the limit distribution of the relative size of the two components in a vector, normalized in a certain way, given that at least one of them is large; see, for instance, \citet[][\S 3]{Kotz:Nada:extr:2000} and \citet[][\S 8--9]{BGST}. The normalization of the components induces a moment constraint on the spectral measure, making its estimation a nontrivial task. 

In the literature, a wide range of approaches has been proposed. \citet[][\S 2--3]{Kotz:Nada:extr:2000} survey many parametric models for the spectral measure, and new models continue to be invented \citep{Cool:Davi:Nave:pair:2010,Bold:Davi:mixt:2007,ballani:schlather:2011}. Here we are mostly concerned  with semiparametric and nonparametric approaches. \cite{Einmahl:2009p327} propose an enhancement of the empirical spectral measure in \cite{Einm:deH:Pite:nonp:2001} by enforcing the moment constraints with empirical likelihood methods. A nonparametric Bayesian method based on the censored-likelihood approach in \cite{Ledf:Tawn:stat:1996} is proposed in \cite{guillotte2011non}.

In this paper we introduce a Euclidean likelihood-based estimator related with the maximum empirical likelihood estimator of \cite{Einmahl:2009p327}. Our estimator replaces the empirical likelihood objective function by the Euclidean distance between the barycenter of the unit simplex and the vector of probability masses of the spectral measure at the observed pseudo-angles \citep{Owen:empi:1991,owen2001empirical,Crpe:Hara:Tres:usin:2009}. This construction allows us to obtain an empirical likelihood-based estimator which is simple and explicitly defined. Its expression is free of Lagrange multipliers, which not only simplifies computations but also leads to a more manageable asymptotic theory. We show that the limit distribution of the empirical process associated with the maximum Euclidean likelihood estimator measure is the same as the one of the maximum empirical likelihood estimator in \cite{Einmahl:2009p327}. Note that standard large-sample results for Euclidean likelihood methods \citep{Xu:larg:1995,Lin:Zhan:bloc:2001} cannot be applied in the context of bivariate extremes.

The paper is organized as follows. In the next section we discuss the probabilistic and geometric frameworks supporting models for bivariate extremes. In Section~\ref{sec:estim} we introduce the maximum Euclidean likelihood estimator for the spectral measure. Large-sample theory is provided in Section~\ref{sec:asym}. Numerical experiments are reported in Section~\ref{numerical.study} and an illustration with extreme temperature data is given in Section~\ref{sec:temp}. Proofs and some details on a smoothing procedure using Beta kernels are given in the Appendix~\ref{app:proofs} and~\ref{app:smooth}, respectively.

\section{Background}
\label{sec:background}

Let $(X_1, Y_1), (X_2, Y_2), \ldots$ be independent and identically distributed bivariate random vectors with continuous marginal distributions $F_X$ and $F_Y$. For the purposes of studying extremal dependence, it is convenient to standardize the margins to the unit Pareto distribution via $X_i^* = 1 / \{ 1 - F_X(X_i) \}$ and $Y_i^* = 1 / \{ 1 - F_Y(Y_i) \}$. Observe that $X_i^*$ exceeds a threshold $t > 1$ if and only if $X_i$ exceeds its tail quantile $F_X^{-1}(1 - 1/t)$; similarly for $Y_i^*$. The transformation to unit Pareto distribution serves to measure the magnitudes of the two components according to a common scale which is free from the actual marginal distributions.

Pickands' representation theorem \citep{Pick:mult:1981} asserts that if the vector of rescaled, componentwise maxima
\[
  \bm{M}_n^* = \biggl( \frac{1}{n} \max_{i=1,\ldots,n} X_i^*, \frac{1}{n}  \max_{i=1\ldots,n} Y_i^* \biggr),
\]
converges in distribution to a non-degenerate limit, then the limiting distribution is a bivariate extreme value distribution $G$ with unit-Fr\'{e}chet margins given by
\begin{equation}
\label{BEV}
  G(x,y)=\exp\bigg\{ -2 \int_{[0, 1]} \max \left( \frac{w}{x},\frac{1-w}{y}  \right) \mathrm{d}H(w)  \bigg\}, \qquad x,y>0.
\end{equation}
The spectral (probability) measure $H$ is a probability distribution on $[0, 1]$ that is arbitrary apart from the moment constraint
\begin{equation}
\label{constraints}
  \int_{[0,1]} w \, \mathrm{d}H(w) = 1/2,
\end{equation}
induced by the marginal distributions $G(z, \infty) = G(\infty, z) = \exp(-1/z)$ for $z > 0$.

The spectral measure $H$ can be interpreted as the limit distribution of the pseudo-angle $W_i = X_i^* / (X_i^* + Y_i^*)$ given that the pseudo-radius $R_i = X_i^* + Y_i^*$ is large. Specifically, weak convergence of $\bm{M}_n^*$ to $G$ is equivalent to
\begin{equation}
\label{eq:H}
  \Pr[ W_i \in \, \cdot \, \mid R_i > t ] \wto H(\,\cdot\,), \qquad t \to \infty.
\end{equation}
The pseudo-angle $W_i$ is close to $0$ or to $1$ if one of the components $X_i^*$ or $Y_i^*$ dominates the other one, given that at least one of them is large. Conversely, the pseudo-angle $W_i$ will be close to $1/2$ if both components $X_i^*$ and $Y_i^*$ are of the same order of magnitude. In case of asymptotic dependence, $G(x, y) = \exp(-1/x-1/y)$, the spectral measure puts mass $1/2$ at the atoms $0$ and $1$, whereas in case of complete asymptotic dependence $G(x, y) = \exp\{-1/\min(x, y)\}$, the spectral measure $H$ reduces to a unit point mass at $1/2$.

Given a sample $(X_1, Y_1), \ldots, (X_n, Y_n)$, we may construct proxies for the unobservable pseudo-angles $W_i$ by setting
\begin{align*}
  \hat{W}_i &= \hat{X}_i^* / ( \hat{X}_i^* + \hat{Y}_i^* ), &
  \hat{R}_i &= \hat{X}_i^* + \hat{Y}_i^*,
\end{align*}
where $\hat{X}_i^* = 1 / \{ 1 - \hat{F}_X(X_i) \}$ and $\hat{Y}_i^* = 1 / \{ 1 - \hat{F}_Y(Y_i) \}$ and where $\hat{F}_X = \hat{F}_{X,n}$ and $\hat{F}_Y = \hat{F}_{Y,n}$ are estimators of the marginal distribution functions $F_X$ and $F_Y$. A robust choice for $\hat{F}_X$ and $\hat{F}_Y$ is the pair of univariate empirical distribution functions, normalized by $n+1$ rather than by $n$ to avoid division by zero. In this case, $\hat{X}_i^*$ and $\hat{Y}_i^*$ are functions of the ranks.

For a high enough threshold $t = t_n$, the collection of angles $\{ \hat{W}_i : i \in K \}$ with $K = K_n = \{ i = 1, \ldots, n : \hat{R}_i > t \}$ can be regarded as if it were a sample from the spectral measure $H$. Parametric or nonparametric inference on $H$ may then be based upon the sample $\{ \hat{W}_i : i \in K \}$. Nevertheless, two complications occur:
\begin{enumerate}
\item The choice of the threshold $t$ comes into play both through the rate of convergence in \eqref{eq:H} and through the effective size $|K| = k$ of the sample of pseudo-angles.
\item The standardization via the estimated margins induces dependence between the pseudo-angles, even when the original random vectors $(X_i, Y_i)$ are independent.
\end{enumerate}

For the construction of estimators of the spectral measure, we may thus pretend that $\{ \hat{W}_i : i \in I \}$ constitutes a sample from $H$. However, for the theoretical analysis of the resulting estimators, the two issues above must be taken into consideration. Failure to do so will lead to a wrong assessment of both the bias and the standard errors of estimators of extremal dependence.

\section{Maximum Euclidean Likelihood Estimator}
\label{sec:estim}

We propose to use Euclidean likelihood methods \citep[pp.~63--66]{owen2001empirical}  to estimate the spectral measure. Let $w_1, \ldots, w_k \in [0, 1]$ be a sample of pseudo-angles, for example the observed values of the random variables $\hat{W}_i$, $i \in K$, in the previous section, with $k = |K|$. The Euclidean loglikelihood $\ell_{\text{E}}$ ratio for a candidate spectral measure $H$ supported on $\{w_1, \ldots, w_k\}$ and assigning probability mass $p_i = H(\{w_i\})$ to $w_i$ is formally defined as 
\[
  \ell_{\text{E}}(\bm{p}) = -\frac{1}{2} \sum_{i=1}^k (k p_i - 1)^2.
\]
The Euclidean likelihood ratio can be viewed as a Euclidean measure of the distance of $\bm{p} = (p_1,\hdots,p_k)$ to the barycenter $(k^{-1},\hdots,k^{-1})$ of the $(k-1)$-dimensional unit simplex. In this sense, the Euclidean likelihood ratio is similar to the empirical loglikelihood ratio 
\[
  \ell(\bm{p}) = \sum_{i=1}^k \log(k p_i),
\]
which can be understood as another measure of the distance from $\bm{p}$ to $(k^{-1},\hdots,k^{-1})$. Note that $\ell_{\text{E}}(\bm{p})$ results from $\ell(\bm{p})$ by truncation of the Taylor expension $\log(1+x) = x - x^2/2 + \cdots$ and the fact that $p_1 + \cdots + p_k = 1$, making the linear term in the expansion disappear.

We seek to maximize $\ell_{\text{E}}(\bm{p})$ subject to the empirical version of the moment constraint \eqref{constraints}. Our estimator $\hat{H}$ for the distribution function of the spectral measure is defined as 
\begin{equation}
\label{eq:mEle:H}
  \hat{H}(w) = \sum_{i=1}^k \hat{p}_i \, I(w_i \le w), \qquad w \in [0, 1],
\end{equation}
the vector of probability masses $\hat{\bm{p}} = (\hat{p}_1, \ldots, \hat{p}_k)$ solving the optimization problem 
\begin{equation}
\label{eq:quadprog}
  \begin{array}{rl}
    \underset{\bm{p} \in \mathbb{R}^k}{\max}        &-\frac{1}{2} \sum_{i=1}^k (k p_i - 1)^2 \\
    \text{s.t.} & \sum_{i=1}^k p_i = 1 \\
                & \sum_{i=1}^k w_i p_i = 1/2.
  \end{array}
\end{equation}
This quadratic program with linear constraints can be solved explicitly with the method of Lagrange multipliers, yielding
\begin{equation}
\label{eq:mEle:p}
  \hat{p}_i = \frac{1}{k} \bigl\{ 1- (\overline{w} - 1/2)S^{-2} (w_i - \overline{w}) \bigr\}, \qquad i = 1,\hdots,k,
\end{equation}
where $\overline{w}$ and $S^2$ denote the sample mean and sample variance of $w_1, \ldots, w_k$, that is,
\begin{align*}
  \overline{w} &= \frac{1}{k} \sum_{i=1}^k w_i, & S^2 &= \frac{1}{k} \sum_{i=1}^k (w_i - \overline{w})^2.
\end{align*}
The weights $\hat{p}_i$ could be negative, but our numerical experiments suggest that this is not as problematic as it may seem at first sight, in agreement with 
\cite{Anto:Bonn:Rena:on:2007} and \cite{Crpe:Hara:Tres:usin:2009}, who claim that the weights $p_i$ are nonnegative with probability tending to one. The second equality constraint in \eqref{eq:quadprog} implies that $\hat{H}$ satisfies the moment constraint \eqref{constraints}, as $\int_{[0, 1]} w \, \mathrm{d} \hat{H}(w) = \sum_i w_i \hat{p}_i = 1/2$, which can be easily verified directly.

The empirical spectral measure estimator of \cite{Einm:deH:Pite:nonp:2001} and the maximum empirical likelihood estimator of \cite{Einmahl:2009p327} can be constructed as in \eqref{eq:quadprog} through suitable changes in the objective function and the constraints. The empirical spectral measure $\dot{H}(w) = \sum_i \dot{p}_i \, I(w_i \le w)$ solves the optimization problem 
\[ \begin{array}{rl}
    \underset{\bm{p} \in \mathbb{R}^k_+}{\max}        &\sum_{i=1}^k
    \log p_i \\
    \mbox{s.t.} & \sum_{i=1}^k p_i = 1.                
    \end{array}
\]
yielding $\dot{p}_i = 1/k$ for $i = 1, \ldots, k$. In contrast, the maximum empirical likelihood estimator $\ddot{H}(w) = \sum_i \ddot{p}_i \, I(w_i \le w)$ has probability masses given by the solution of
\begin{equation} \label{optim_mele}
  \begin{array}{rl}
    \underset{\bm{p} \in \mathbb{R}^k_+}{\max}        &\sum_{i=1}^k \log p_i \\
    \mbox{s.t.} & \sum_{i=1}^k p_i = 1 \\
                & \sum_{i=1}^k w_i p_i = 1/2.
    \end{array}
\end{equation}
By the method of Lagrange multipliers, the solution is given by
\[\label{tilde.p_i}
  \ddot{p}_i = \frac{1}{k} \frac{1}{1+ \lambda (w_i-1/2)}, \qquad i = 1,\hdots,k,
\]
where $\lambda \in \mathbb{R}$ is the Lagrange multiplier associated to the second equality constraint in \eqref{optim_mele}, defined implicitly as the solution to the equation
\[
  \frac{1}{k} \sum_{i=1}^k \frac{w_i-1/2}{{1+ \lambda (w_i-1/2)}}=0,
\]
see also \cite{Qin:Lawl:empi:1994}.

\section{Large-Sample Theory}
\label{sec:asym}

The maximum Euclidean likelihood estimator $\hat{H}$ in \eqref{eq:mEle:H} can be expressed in terms of the empirical spectral measure $\dot{H}$ given by
\[ 
  \dot{H}(w) = \frac{1}{k} \sum_{i=1}^k I(w_i \le w), \qquad w \in [0, 1]. 
\]
Indeed, $\overline{w}$ and $S^2$ are just the mean and the variance of $\dot{H}$, and the expression \eqref{eq:mEle:p} for the weights $\hat{p}_i$ can be written as
\[
  \frac{\mathrm{d} \hat{H}}{\mathrm{d} \dot{H}}(v) = 1 - (\overline{w} - 1/2) \, S^{-2} \, (v - \overline{w}), \qquad v \in [0, 1].
\]
Integrating out this `likelihood ratio' over $v \in [0, w]$ yields the identity
\[
  \hat{H}(w) = (\Phi(\dot{H}))(w), \qquad w \in [0, 1],
\]
where the transformation $\Phi$ is defined as follows. Let $\mathbb{D}_\Phi$ be the set of cumulative distribution functions of non-degenerate probability measures on $[0, 1]$. For $F \in \mathbb{D}_\Phi$, the function $\Phi(F)$ on $[0, 1]$ is defined by
\[
  (\Phi(F))(w) = F(w) - (\mu_F - 1/2) \, \sigma_F^{-2} \, \int_{[0, w]} (v - \mu_F) \, \mathrm{d} F(v), \qquad w \in [0, 1].
\]
Here $\mu_F = \int_{[0, 1]} v \, \mathrm{d} F(v)$ and $\sigma_F^2 = \int_{[0, 1]} (v - \mu_F)^2 \, \mathrm{d} F(v)$ denote the mean and the (non-zero) variance of $F$. 

We view $\mathbb{D}_\Phi$ as a subset of the Banach space $\ell^\infty([0, 1])$ of bounded, real-valued functions on $[0, 1]$ equipped with the supremum norm $\| \,\cdot \, \|_\infty$. The map $\Phi$ takes values in $\ell^\infty([0, 1])$ as well. Weak convergence in $\ell^\infty([0, 1])$ is denoted by the arrow `$\dto$' and is to be understood as in \cite{vdv:wellner:1996}. 

Asymptotic properties of the empirical spectral measure together with smoothness properties of $\Phi$ lead to asymptotic properties of the maximum Euclidean likelihood estimator:
\begin{itemize} 
\item Continuity of the map $\Phi$ together with consistency of the empirical spectral measure yields consistency of the maximum Euclidean likelihood estimator (continuous mapping theorem).
\item Hadamard differentiability of the map $\Phi$ together with asymptotic normality of the empirical spectral measure yields asymptotic normality of the maximum Euclidean likelihood estimator (functional delta method).
\end{itemize}

The following theorems are formulated in terms of maps $\dot{H}_n$ taking values in $\mathbb{D}_\Phi$. The case to have in mind is the empirical spectral measure $\dot{H}_n(w) = k_n^{-1} \sum_{i \in K_n} I( \hat{W}_i \le w )$ with $\{ \hat{W}_i : i \in K_n \}$ and $k_n = |K_n|$ as in Section~\ref{sec:background}. In Theorem~3.1 and equation~(7.1) of \cite{Einmahl:2009p327}, asymptotic normality of $\dot{H}_n$ is established under certain smoothness conditions on $H$ and growth conditions on the threshold sequence $t_n$.

\begin{theorem}[Consistency]
\label{thm:consistency}
If $\dot{H}_n$ are maps taking values in $\mathbb{D}_\Phi$ and if $\| \dot{H}_n - H \|_\infty \to 0$ in outer probability for some nondegenerate spectral measure $H$, then, writing $\hat{H}_n = \Phi( \dot{H}_n )$, we also have $\| \hat{H}_n - H \|_\infty \to 0$ in outer probability.
\end{theorem}

The proof of this and the next theorem is given in Appendix~\ref{app:proofs}. In the next theorem, the rate sequence $r_n$ is to be thought of as $\sqrt{k_n}$. Let $\mathcal{C}([0, 1])$ be the space of continuous, real-valued functions on $[0, 1]$.

\begin{theorem}[Asymptotic normality]
\label{thm:an}
Let $\dot{H}_n$ and $H$ be as in Theorem~\ref{thm:consistency}. If $H$ is continuous and if
\[
  r_n (\dot{H}_n - H ) \dto \beta, \qquad n \to \infty,
\]
in $\ell^\infty([0, 1])$, with $0 < r_n \to \infty$ and with $\beta$ a random element of $\mathcal{C}([0, 1])$, then also
\[ 
  r_n ( \hat{H}_n - H) \dto \gamma, \qquad n \to \infty
\]
with
\begin{equation}
\label{eq:gamma}
  \gamma(w) = \beta(w) - \sigma_H^{-2} \, \int_0^1 \beta(v) \, \mathrm{d}v \, \int_0^w (1/2 - v) \, \mathrm{d}H(v), \qquad w \in [0, 1].
\end{equation}
\end{theorem}

Comparing the expression for $\gamma$ in \eqref{eq:gamma} with the one for $\gamma$ in (4.7) in \cite{Einmahl:2009p327}, we see that the link between the processes $\beta$ and $\gamma$ here is the same as the one between the processes $\beta$ and $\gamma$ in \cite{Einmahl:2009p327}. It follows that tuning the empirical spectral measure via either maximum empirical likelihood or maximum Euclidean likelihood makes no difference asymptotically. The numerical experiments below confirm this asymptotic equivalence. To facilitate comparisons with \cite{Einmahl:2009p327}, note that our pseudo-angle $w \in [0, 1]$ is related to their radial angle $\theta \in [0, \pi/2]$ via $\theta = \arctan \{ w / (1-w) \}$, and that the function $f$ in (4.2) in \cite{Einmahl:2009p327} reduces to $f(\theta) = (\sin \theta - \cos \theta)/(\sin \theta + \cos \theta) = 2w - 1$.

How does the additional term $\sigma_H^{-2} \, \int_0^1 \beta(v) \, \mathrm{d}v \, \int_0^w (1/2 - v) \, \mathrm{d}H(v)$ influence the asymptotic distribution of the maximum Euclidean/empirical estimator? Given the complicated nature of the covariance function of the process $\beta$, see (3.7) and (4.7) in \cite{Einmahl:2009p327}, it is virtually impossible to draw any conclusions theoretically. However, Monte Carlo simulations in \citet[][\S 5.1]{Einmahl:2009p327} confirm that the maximum empirical likelihood estimator is typically more efficient than the ordinary empirical spectral measure. These findings are confirmed in the next section.

\section{Monte Carlo Simulations} 
\label{numerical.study} 

In this section, the maximum Euclidean likelihood estimator is compared with the empirical spectral measure and the maximum empirical likelihood estimator by means of Monte Carlo simulations. The comparisons are made on the basis of the mean integrated squared error,
\[
  \text{MISE($\cdot$)} = E \bigg[\int_{0}^{1}\{\cdot-H(w)\}^{2}\mathrm{d}w\bigg].
\]

The bivariate extreme value distribution $G$ with logistic dependence structure is defined by
\[
  G(x,y) = \exp \{ -(x^{-1/\alpha}+y^{-1/\alpha})^{\alpha} \}, \qquad x,y>0,
\]
in terms of a parameter $\alpha \in (0, 1]$. Smaller values of $\alpha$ yield stronger dependence and the limiting cases $\alpha \to 0$ and $\alpha = 1$ correspond to complete dependence and exact independence, respectively. For $0 < \alpha < 1$, the spectral measure $H_\alpha$ is absolutely continuous with density
\[
  \frac{\mathrm{d}H_\alpha(w)}{\mathrm{d}w} 
  = \frac{1}{2}(\alpha^{-1}-1)\{w(1-w)\}^{-1-1/\alpha}\big\{w^{-1/\alpha}+(1-w)^{-1/\alpha}\big\}^{\alpha-2},
  \qquad 0 < w < 1.
\]
Here we consider $\alpha_{\textsc{i}}=0.8$ and
$\alpha_{\textsc{ii}}=0.4$, with stronger extremal dependence
corresponding to case \textsc{ii}. For each of these two models, 1000
Monte Carlo samples of size 1000 were generated. The thresholds were
set at the empirical quantiles of the radius $\widehat{R}$ given by
$t=75\%,75.5\%,\ldots,99.5\%$ for case \textsc{i} and
$t=50\%,50.5\%,\ldots,99.5\%$ for case \textsc{ii}. 
The margins were estimated parametrically, by fitting univariate extreme value distributions using maximum likelihood.

In Figure~\ref{fig:trajectory}, a typical trajectory of the estimators is shown, illustrating the closeness of the maximum Euclidean and empirical likelihood estimators. The good performance of the maximum Euclidean/empirical spectral measure is confirmed by Figure~\ref{fig:mise2}. For larger $k$ (lower threshold $t$), the bias coming from the approximation error in \eqref{eq:H} is clearly visible.

\begin{figure}
\begin{center}
\begin{tabular}{c}
\includegraphics[width=0.45\textwidth]{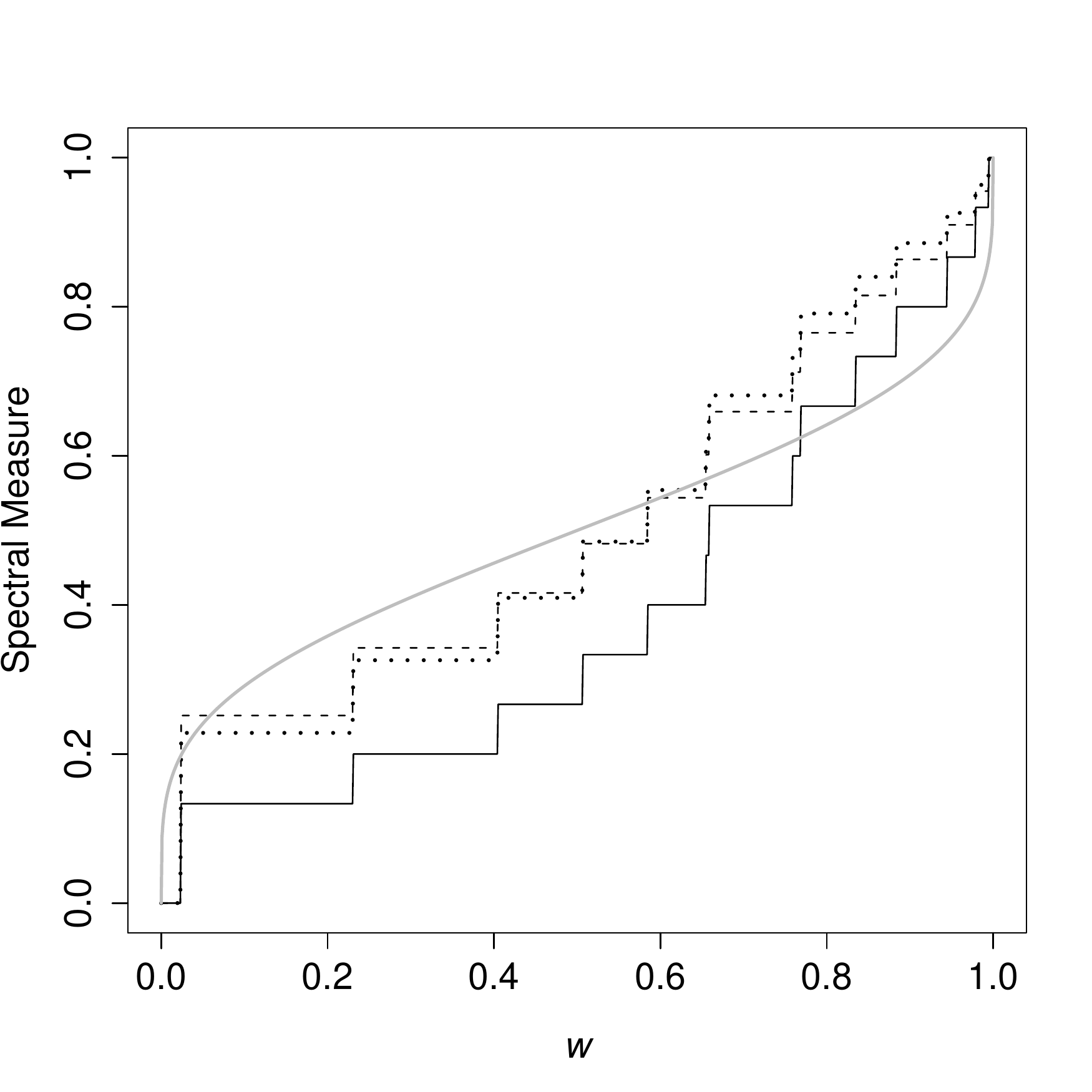} 
\includegraphics[width=0.45\textwidth]{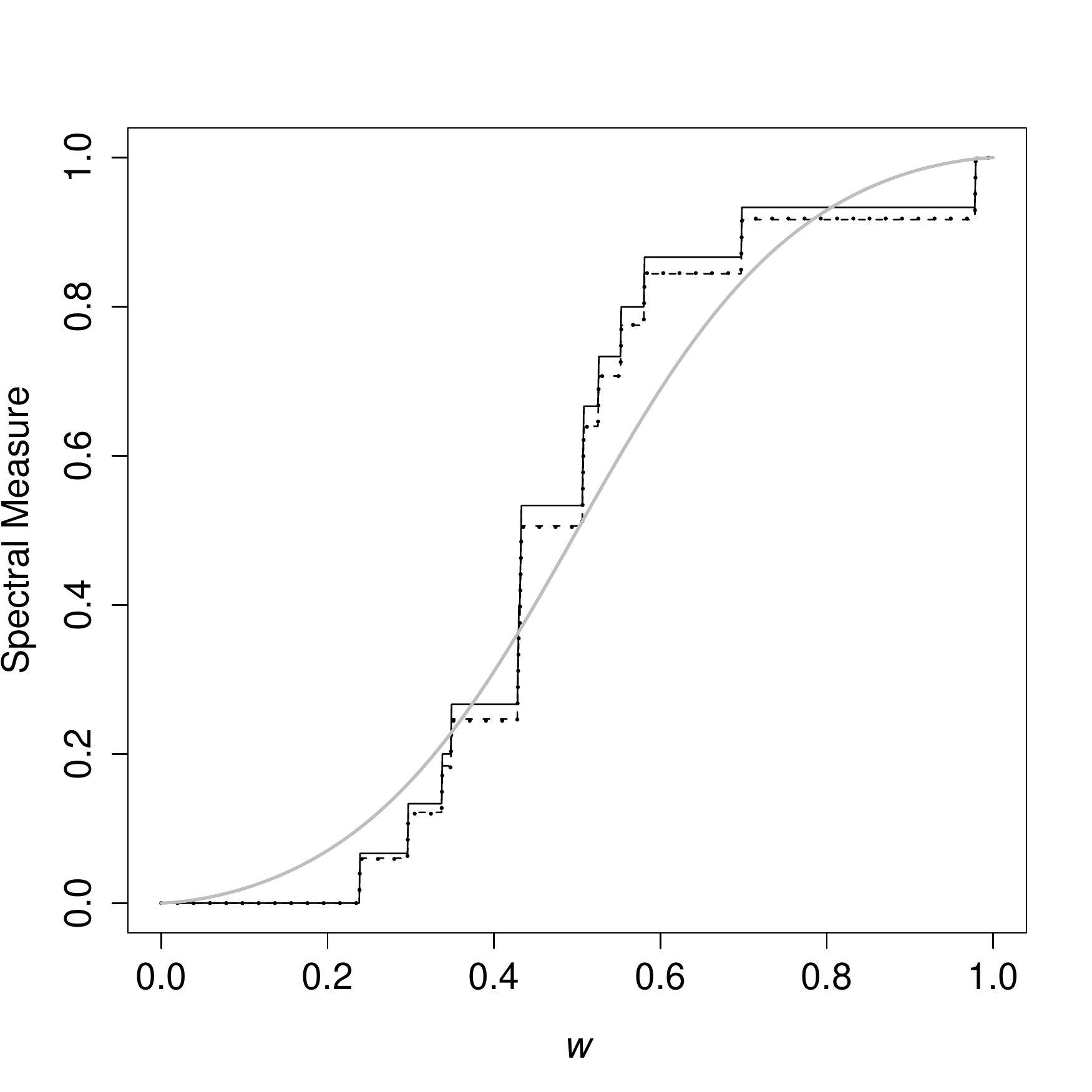} 
\end{tabular}
\caption{\small Examples of trajectories of the empirical spectral measure (solid), the maximum empirical likelihood estimator (dashed), and the maximum Euclidean likelihood estimator (dotted). The solid grey line corresponds to the true spectral measure $H_\alpha$ coming from the bivariate logistic extreme value distribution with parameters $\alpha_{\textsc{i}} = 0.8$ (left) and $\alpha_{\textsc{ii}} = 0.4$ (right). \label{fig:trajectory}}
\end{center}
\end{figure} 

\begin{figure}
\begin{center}
\begin{tabular}{cc}
\includegraphics[width=0.45\textwidth]{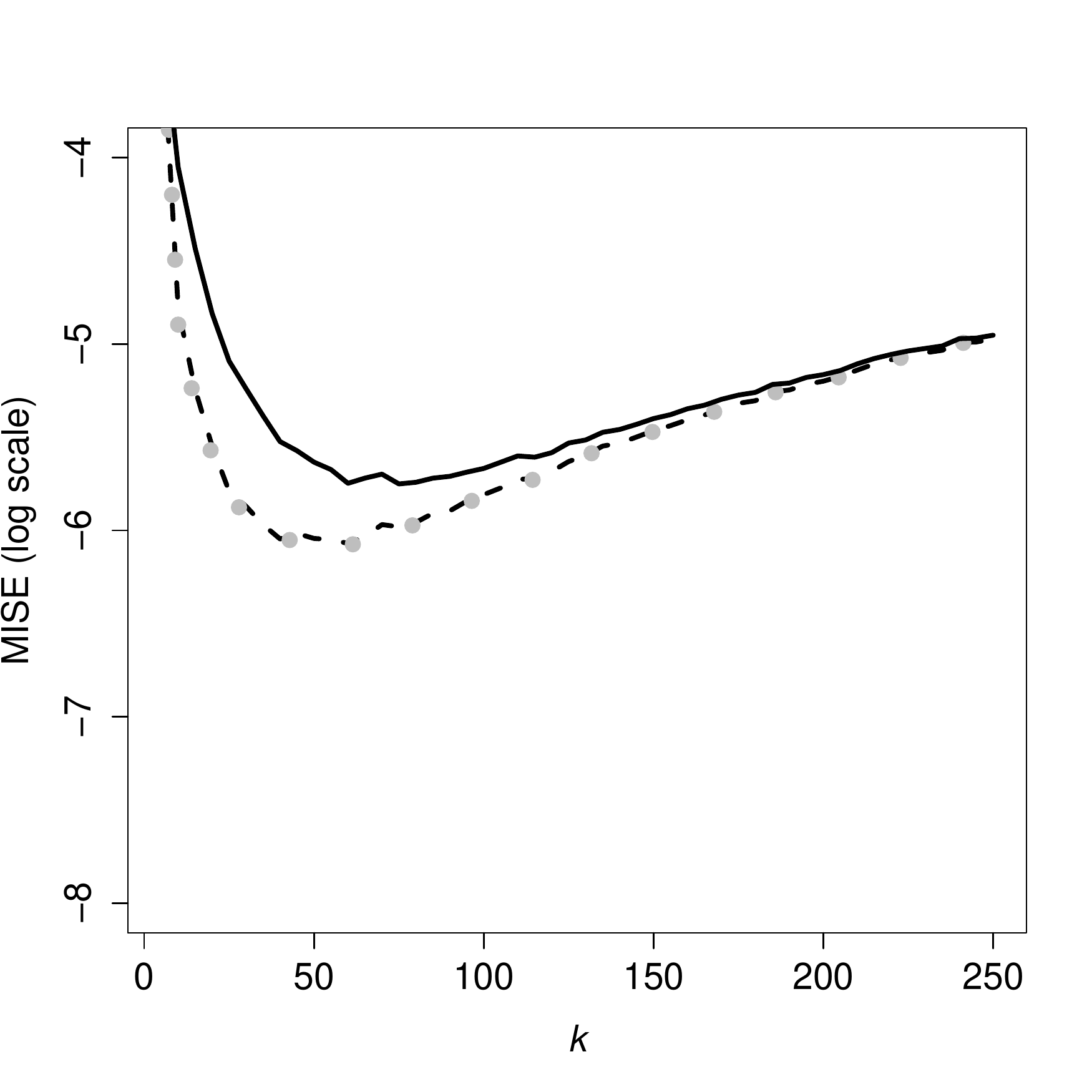}&
\includegraphics[width=0.45\textwidth]{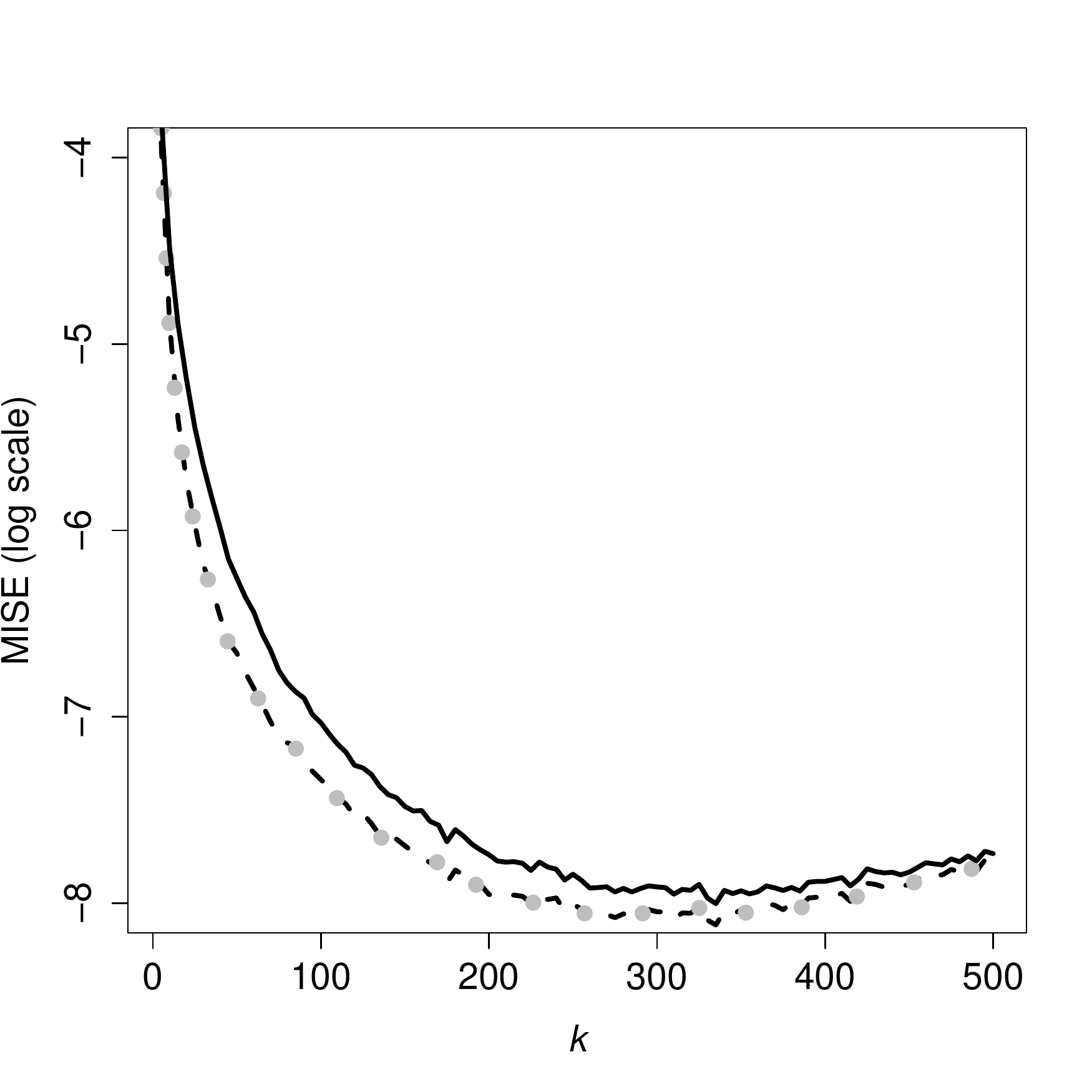}
\end{tabular}
\caption{\small Logarithms of the mean integrated squared errors of the spectral measure estimates based on 1000 samples of size 1000 from the bivariate logistic extreme value distribution with parameters $\alpha_{\textsc{i}} = 0.8$ (left) and $\alpha_{\textsc{ii}} = 0.4$ (right). The solid, dashed, and dotted lines correspond to the empirical spectral measure, the maximum empirical likelihood estimator, and the maximum Euclidean likelihood estimator, respectively. \label{fig:mise2}}
\end{center}
\end{figure} 

Numerical experiments in \citet[][\S 5.2]{Einmahl:2009p327} show that the presence of atoms at the endpoints $0$ and $1$ has an adverse effect on maximum
Euclidean/empirical likelihood estimates, and this finding is further confirmed by \citet[][\S 7.1]{guillotte2011non}. Indeed, by construction, the pseudo-angles $\hat{W}_i$ will never be exactly $0$ or $1$. The empirical spectral measure therefore does not assign any mass at $0$ and $1$, and this situation cannot be remedied by the maximum empirical or Euclidean likelihood estimators, having the same support as the empirical spectral measure.

The weights $\hat{p}_i$ of the maximum Euclidean likelihood estimator
can be negative. However, as can be seen from Figure
\ref{fig:proportion}, the weights tend to be positive overall, except
for extremely high thresholds, with the proportion of negative weights
being smaller in case \textsc{i}. This suggests that the closer we
get to exact independence, the lower the proportion of negative weights.

\begin{figure}[H]
  \begin{center}
  \begin{tabular}{c}
    \includegraphics[width=0.45\textwidth]{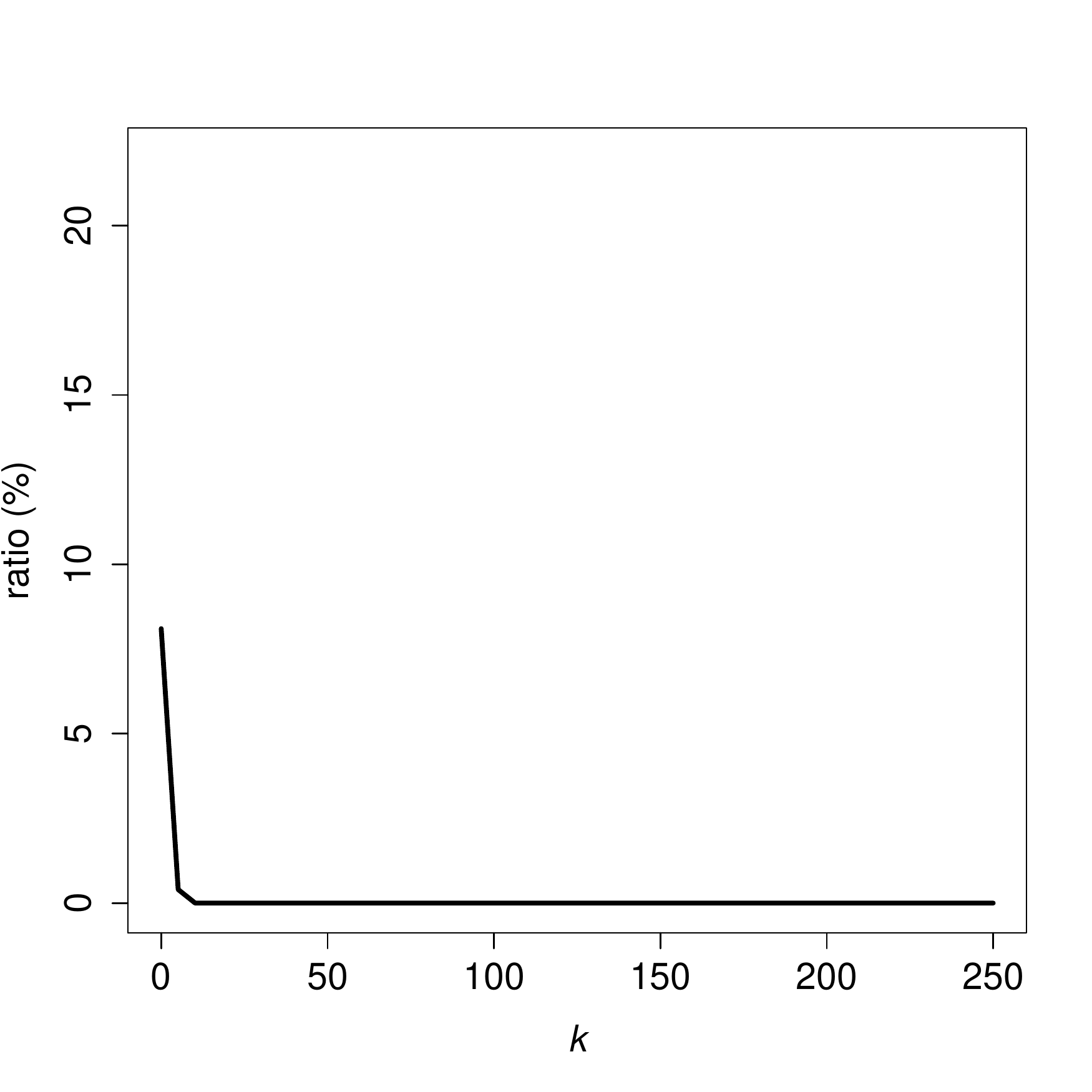} 
    \includegraphics[width=0.45\textwidth]{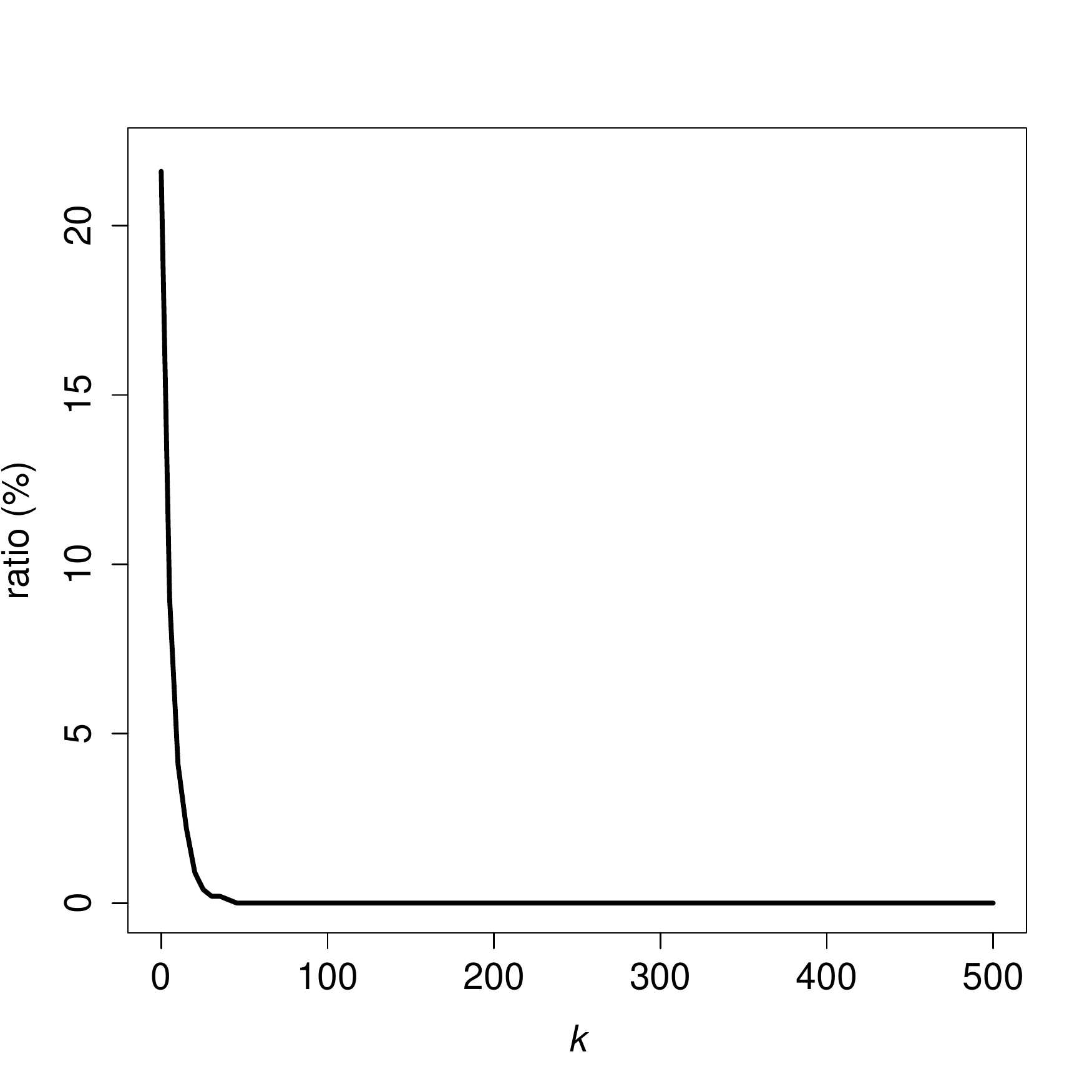} 
  \end{tabular}
  \caption{\small Proportion of negative weights based on 1000 samples of size 1000 
    from the bivariate logistic extreme value distribution with
    parameters $\alpha_{\textsc{i}} = 0.8$ (left) and $\alpha_{\textsc{ii}} = 0.4$ (right). \label{fig:proportion}}
\end{center}
\end{figure}

\section{Extreme Temperature Data Analysis}
\label{sec:temp}

\subsection{Data Description and Preliminary Considerations}
The data were gathered from the Long-term Forest Ecosystem Research
database, which is maintained by LWF (Langfristige Wald\"okosystem-Forschung), and consist of daily average meteorological measurements made in
Beatenberg's forest in the canton of Bern, 
Switzerland. More information on these data can be found at

\begin{verbatim}
http://www.wsl.ch
\end{verbatim}

\noindent and for an extensive study see
\cite{Ferrez:Davison:Rebetez:2011}. Two time series 
of air temperature data are available:
One in the open field and the other
in a nearby site under the forest cover.
Our aim is to understand how the extremes in the open relate with
those under the canopy; comparison of open-site and below-canopy
climatic conditions is a subject of considerable interest 
in Forestry and Meteorology \cite[][]{Renaud:Rebetez:2009, Ferrez:Davison:Rebetez:2011, Renaud:2011}.
The raw data are plotted in Figure~\ref{raw}, but before we are able to measure extremal dependence of open air and
forest cover temperatures we first need to preprocess the data. The preprocessing step is
the same as in \citet[][\S 3.1]{Ferrez:Davison:Rebetez:2011} and further
details can be found in there.

\begin{figure}[H]
\begin{center}
\begin{tabular}{cc}
\includegraphics[width=.45\textwidth]{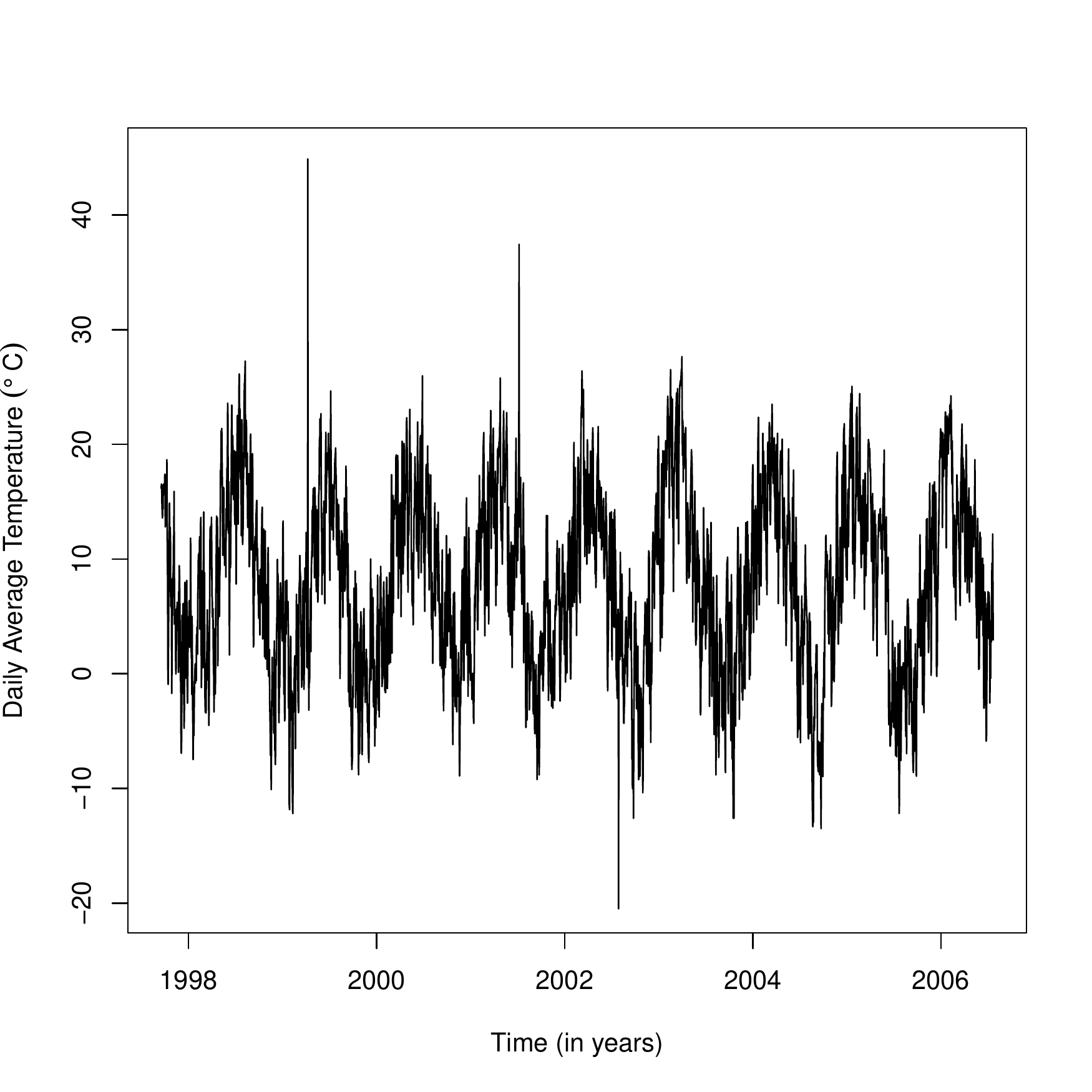}&
\includegraphics[width=.45\textwidth]{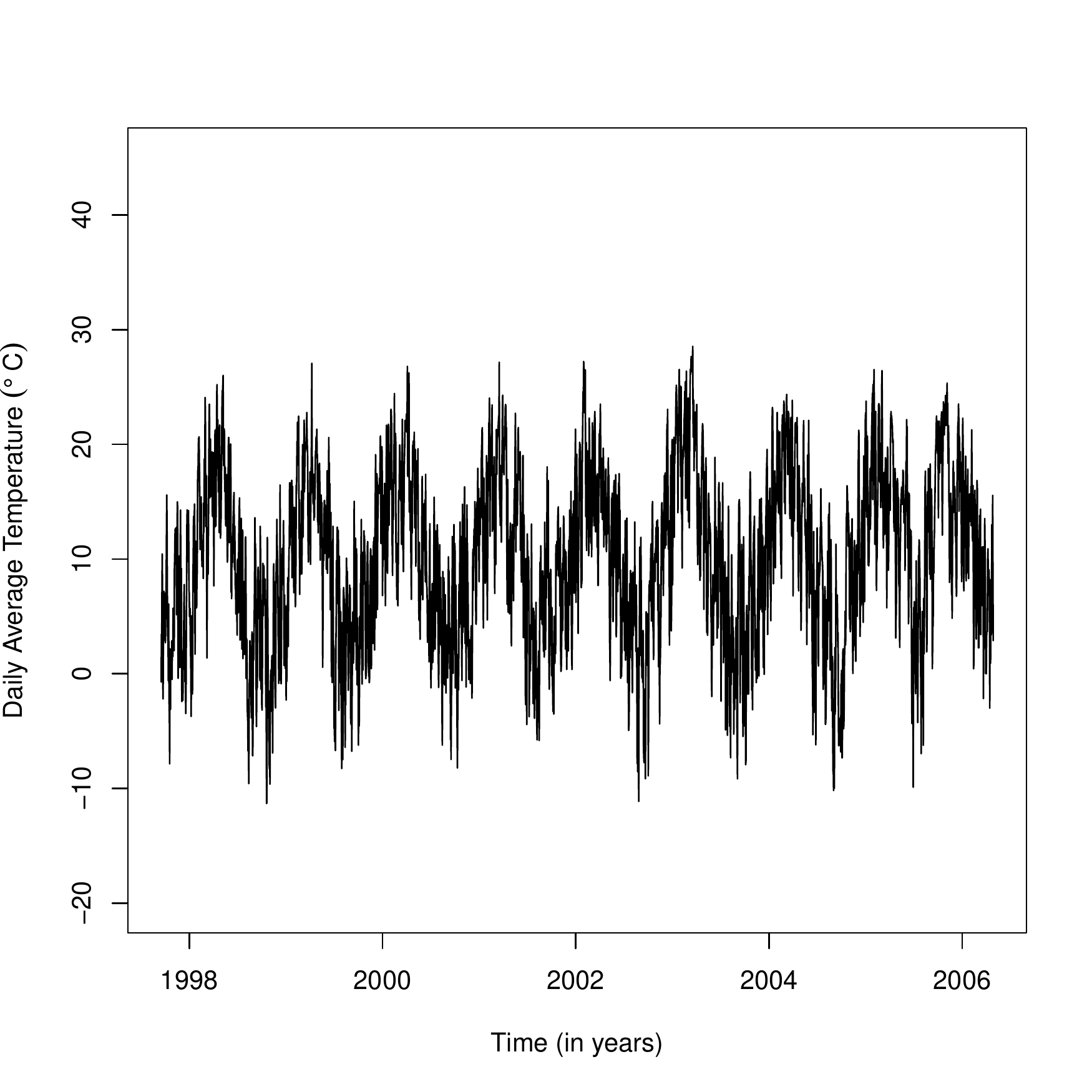} 
\end{tabular}
\caption{\small Daily average air temperatures from the meteorological station in Beatenberg's forest: under the forest cover (left) and in the open field (right). \label{raw}}
\end{center}
\end{figure} 

We consider daily maxima of the residual series that result from removal
of the annual cycle in both location and scale, and we then take
the residuals at their 98\% quantile; hence the threshold boundary is
defined as $U = \{(x,y)\in[0,\infty)^2: x+y =
\widehat{F}^{-1}_{\widehat{R}}(0.98) = 105.83\}$, so that there are $|U|=k=57$
exceedances. 

\begin{figure}[H]
\centering
\includegraphics[width=9cm]{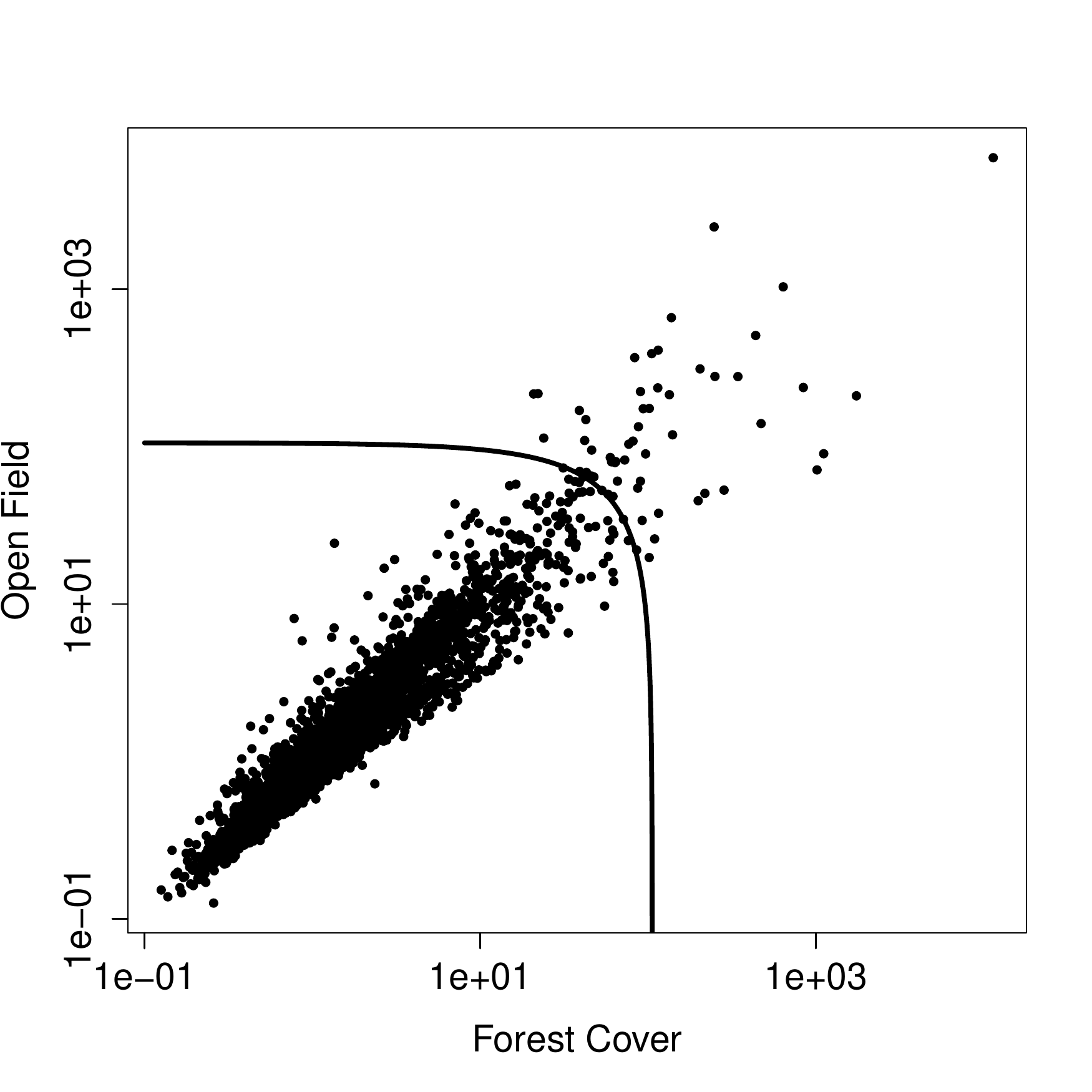} 
\caption{\small Scatterplot of air temperature data after transformation to unit Fr\'echet scale; the solid line corresponds to the boundary threshold in the log-log scale, with both axes being logarithmic. \label{open.vs.cover}}
\end{figure} 

\begin{figure}[H]
\begin{center}
\begin{tabular}{cc}
\includegraphics[width=.45\textwidth]{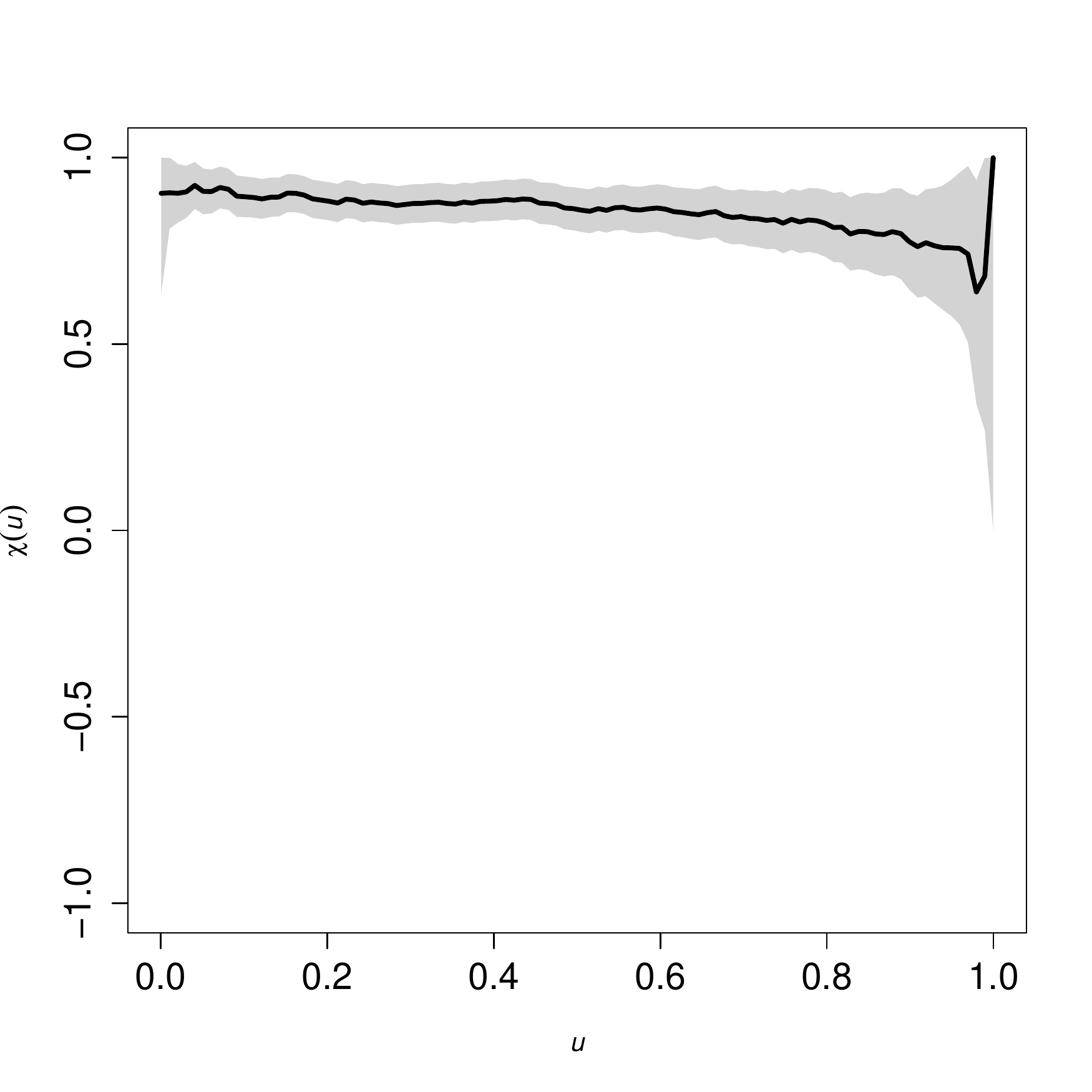}& 
\includegraphics[width=.45\textwidth]{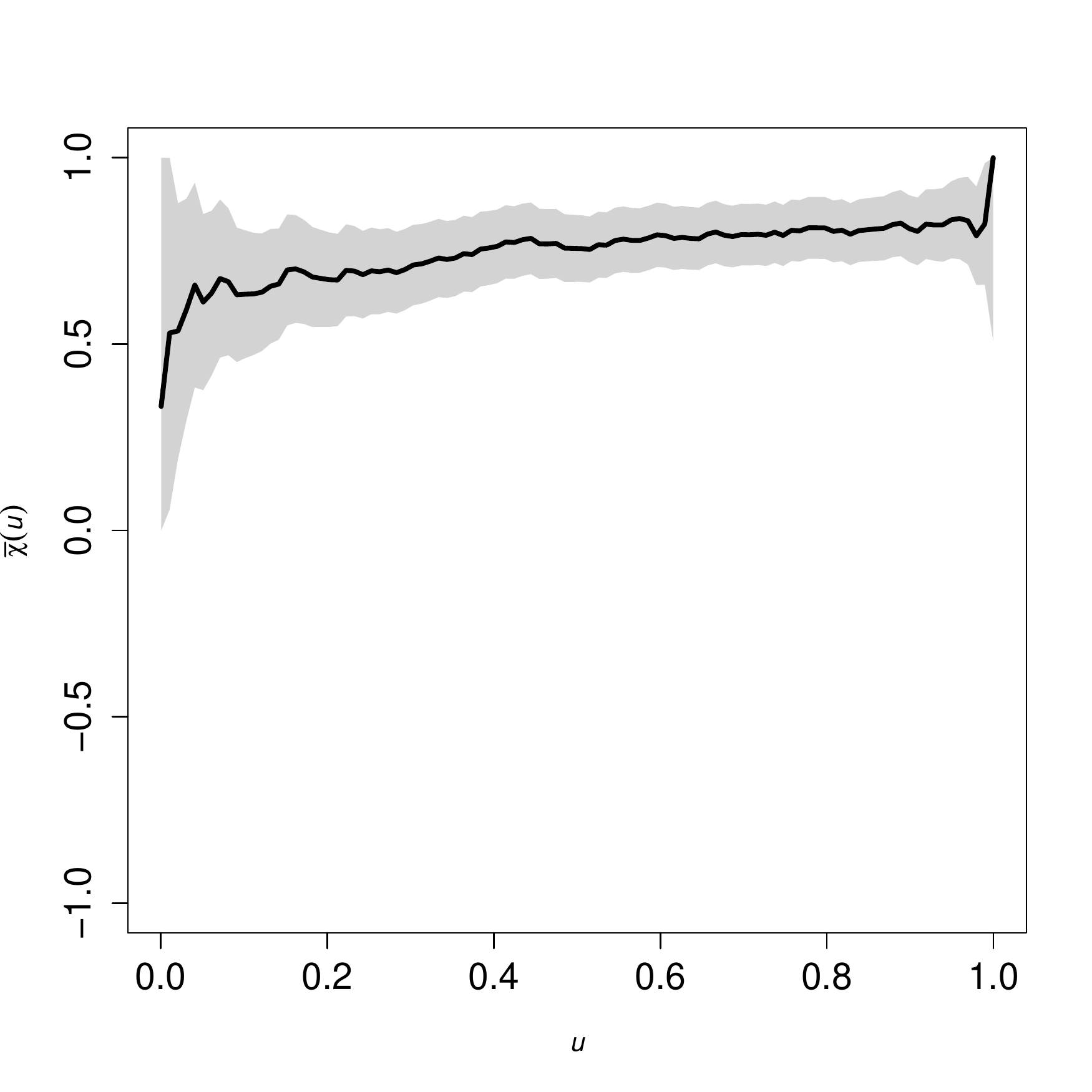} 
\end{tabular}
\caption{\small Empirical estimates and 95\% pointwise confidence intervals for $\chi(u)$ (left) and $\overline{\chi}(u)$ (right) as a function of $u \in (0,1)$. \label{chis}}
\end{center}
\end{figure} 

The dependence between open field and forest cover
temperatures can be observed in Figure~\ref{open.vs.cover}, where we
plot a log-log scale scatterplot of the unit Fr\'echet data, and where
we note that after log transformation the linearity of the threshold boundary
$U$ is perturbed.

Spectral measures are only appropriate for modeling asymptotically
dependent data. Although this issue has been already
addressed in \citet[][]{Ferrez:Davison:Rebetez:2011}, for
exploratory purposes we present in Figure~\ref{chis} the
empirical estimates of the dependence coefficients $\chi(u)$ and $\overline{\chi}(u)$, for $0 < u < 1$, defined in \cite{coles:heffernan:tawn:1999} as
\begin{align*}
  \chi(u) &= 2-\log \frac{\Pr[ F_X(X) < u, F_Y(Y) < u ]}{\log \Pr[ F_X(X) < u ]}, &
  \overline{\chi}(u) &= \frac{2\log(1-u)}{\log \Pr[F_X(X)>u, F_Y(Y)>u]} - 1.
\end{align*}
Although these plots fail to have a clear-cut interpretation given the large uncertainty entailed in the estimation, the point estimates seem to be consistent with asymptotic dependence as already noticed by \citet[][]{Ferrez:Davison:Rebetez:2011}.

\subsection{Extremal Dependence of Open Air and Forest Cover Temperatures}

We now apply the maximum Euclidean likelihood estimator to measure extremal dependence of open air and forest cover temperatures. The estimated spectral measure is shown in Figure~\ref{fig:estimates1}. All weights are positive, i.e.~$\hat{p}_i>0$, for $i=1,\hdots,57$.

By construction, the estimate of the spectral measure is discrete. A smooth version which still obeys the moment constraint \eqref{constraints} can easily obtained by smoothing the maximum Euclidean or empirical likelihood estimator with a Beta kernel. Related ideas are already explored in \cite{Hall:Presnell:1999} and \cite{Chen:1997}. Details are given in Appendix~\ref{app:smooth}. 

A cross-validatory procedure was used to select the bandwidth, yielding a concentration parameter of $\nu \approx 163$. Numerical experiments in \cite{Warchol2012} suggest that convoluting empirical likelihood-based estimators with a Beta kernel yields a further reduction in mean integrated squared error. The Beta kernel even outperforms Chen's kernel \citep{Chen:1999}, which is asymptotically optimal under some conditions \citep{Bouezmarni:Rolin:mixt:2003}, but which is unable to conserve the moment constraint. 

From the smoothed spectral measure, we obtain an estimate of the spectral density and plug-in estimators for the Pickands dependence function $A(w) = 1 - w + 2 \int_0^w H(v) \, \mathrm{d}v$, $w \in [0, 1]$, and the bivariate extreme value distribution in \eqref{BEV}. The estimated spectral density is compared with the fit obtained from the asymmetric logistic model
\begin{multline*}
  H_{\alpha,\psi_1,\psi_2}(w) = \frac{1}{2}
  \bigg[1+\psi_1+\psi_2-\{\psi_1^{1/\alpha}(1-w)^{1/\alpha-1}-\psi_2^rw^{1/\alpha-1}\} \\
  \{\psi_1^{1/\alpha} (1-w)^{1/\alpha}+\psi_2^{1/\alpha}w^{1/\alpha}\}^{1/1/\alpha-1}\bigg], \qquad w \in [0, 1], 
\end{multline*}
with parameter estimates $\widehat{\psi}_1 = 0.78$ (standard error $0.03$), $\widehat{\psi}_2 = 0.90$ ($0.03$) and $\widehat{\alpha} =  0.30$ ($0.02$). The asymmetric logistic model was considered by \cite{Ferrez:Davison:Rebetez:2011} as the parametric model that achieved the ``best overall fit.''

In Figure~\ref{fig:estimates1} we also plot the smooth spectral
measure and corresponding spectral density which are obtained by suitably convoluting the
empirical Euclidean spectral measure with a Beta kernel as described in \eqref{h.smooth} and \eqref{H.smooth}. Since more mass concentrated over 1/2 corresponds to more extremal dependence, and more mass concentrated on 0 and 1 corresponds more independence in the extremes, a rough interpretation for our context is as follows: The lower the shelter ability of the forest, the more mass should be concentrated around 1/2, whereas higher shelter ability corresponds to the case where the spectral measure
gets more mass concentrated at 1; more mass concentrated at 0 suggests relatively more extreme events under the forest cover, suggesting that the forest has the ability to retain heat during extreme events.

\begin{figure}
\begin{center}
\begin{tabular}{cc}
\includegraphics[width=.45\textwidth]{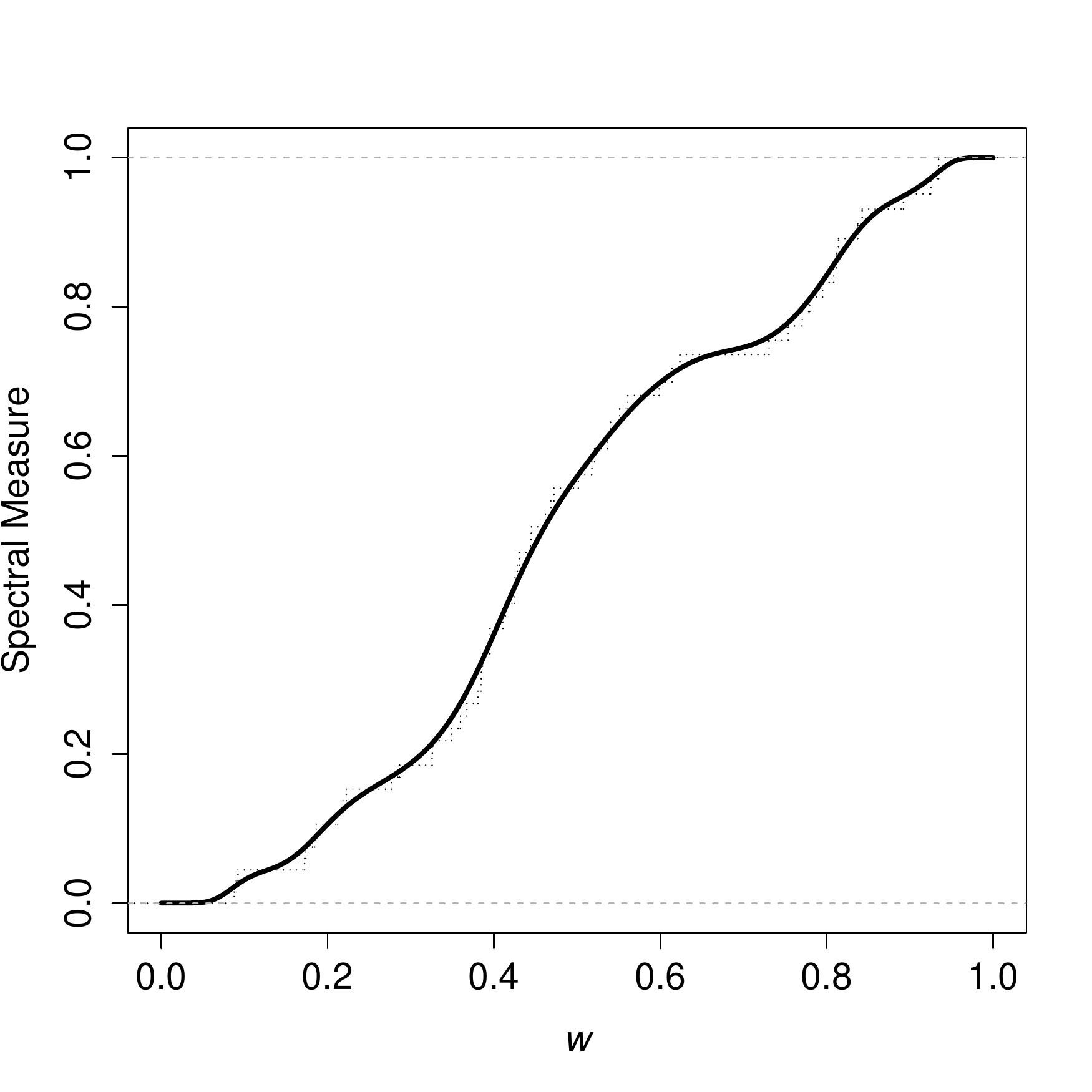}&
\includegraphics[width=.45\textwidth]{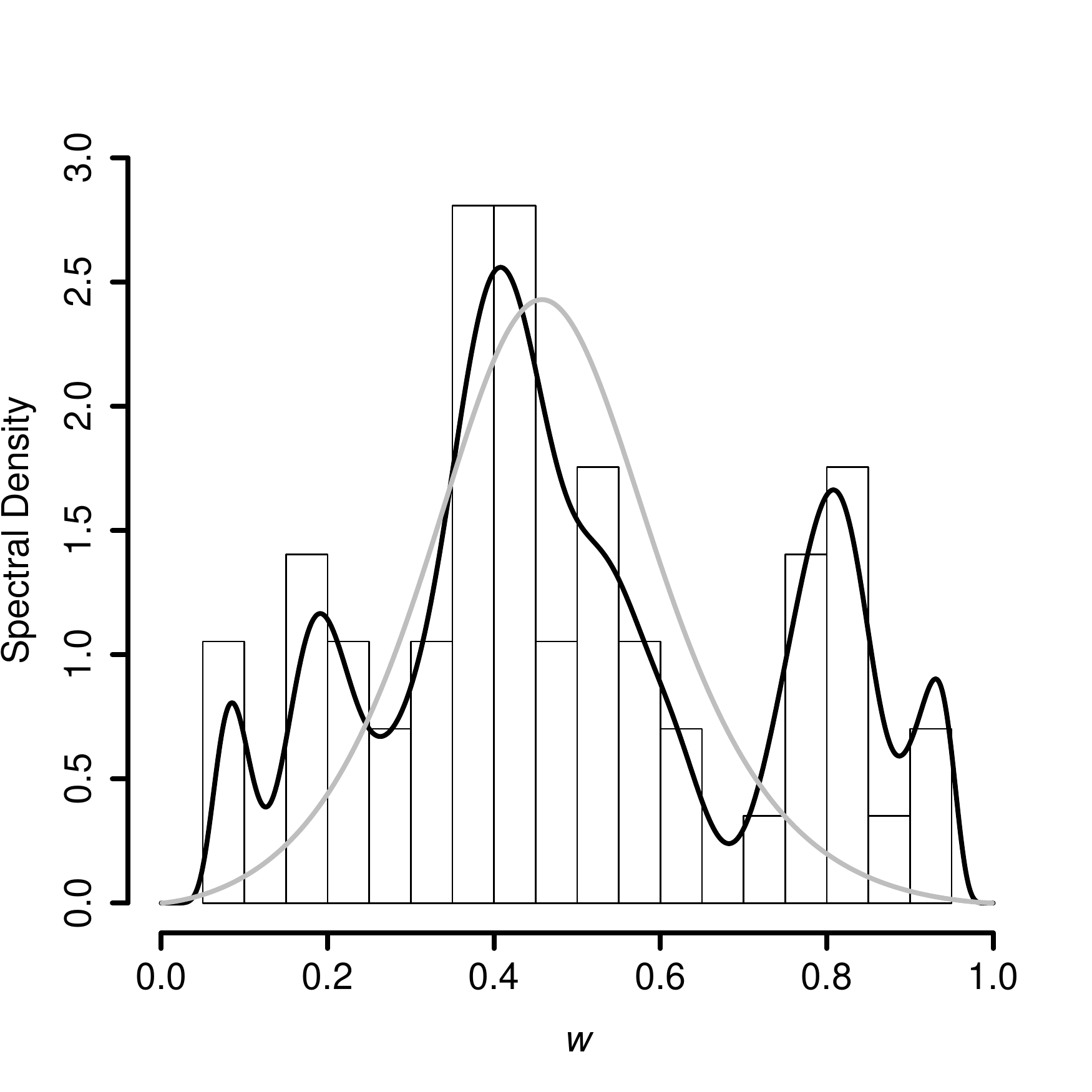}
\end{tabular}
\caption{\small Estimates of the spectral measure (left) and the spectral density (right). Left: the dotted line corresponds to the empirical Euclidean 
    spectral measure, and the solid line corresponds to its smooth version constructed using \eqref{H.smooth}. Right: The solid line corresponds to the
    smooth spectral density obtained from the Euclidean likelihood weights using \eqref{h.smooth}, and the gray line represents the fit from the asymmetric logistic model with $(\widehat{\psi}_1,\widehat{\psi}_2)=(0.78,0.90)$ and $\widehat{\rho} =  0.30$. \label{fig:estimates1}}
\end{center}
\end{figure} 

In Figure~\ref{fig:estimates2}.1 we plot the corresponding Pickands dependence function. More extremal dependence corresponds to lower Pickands dependence functions, and the deeper these are on the right the less frequent are the extreme events under the forest cover relatively to the open field. Our analysis suggests that extreme high temperatures under the forest cover are more frequent than expected from a corresponding parametric analysis. This somewhat surprising finding is already predicted in \citet[Fig.~4]{Ferrez:Davison:Rebetez:2011}. The phenomenon may be due to the ability of some forests to retain heat, acting like a greenhouse, or it may be connected with the way that other features of the forest's structure can alter its microclimate \citep{Renaud:2011}. Along with the Pickands dependence function, we also plot in Figure~\ref{fig:estimates2}.1 the pseudo-angles which provide further evidence of a marked right skewness. 

The joint behavior of temperatures in the open and under forest cover can also be examined from the estimated bivariate extreme value distribution function plotted in Figure~\ref{fig:estimates2}.2, which was constructed by convoluting the empirical Euclidean spectral measure with a Beta kernel as described in \eqref{G.smooth}.

\begin{figure}
\begin{center}
\begin{tabular}{cc}
\includegraphics[width=.45\textwidth]{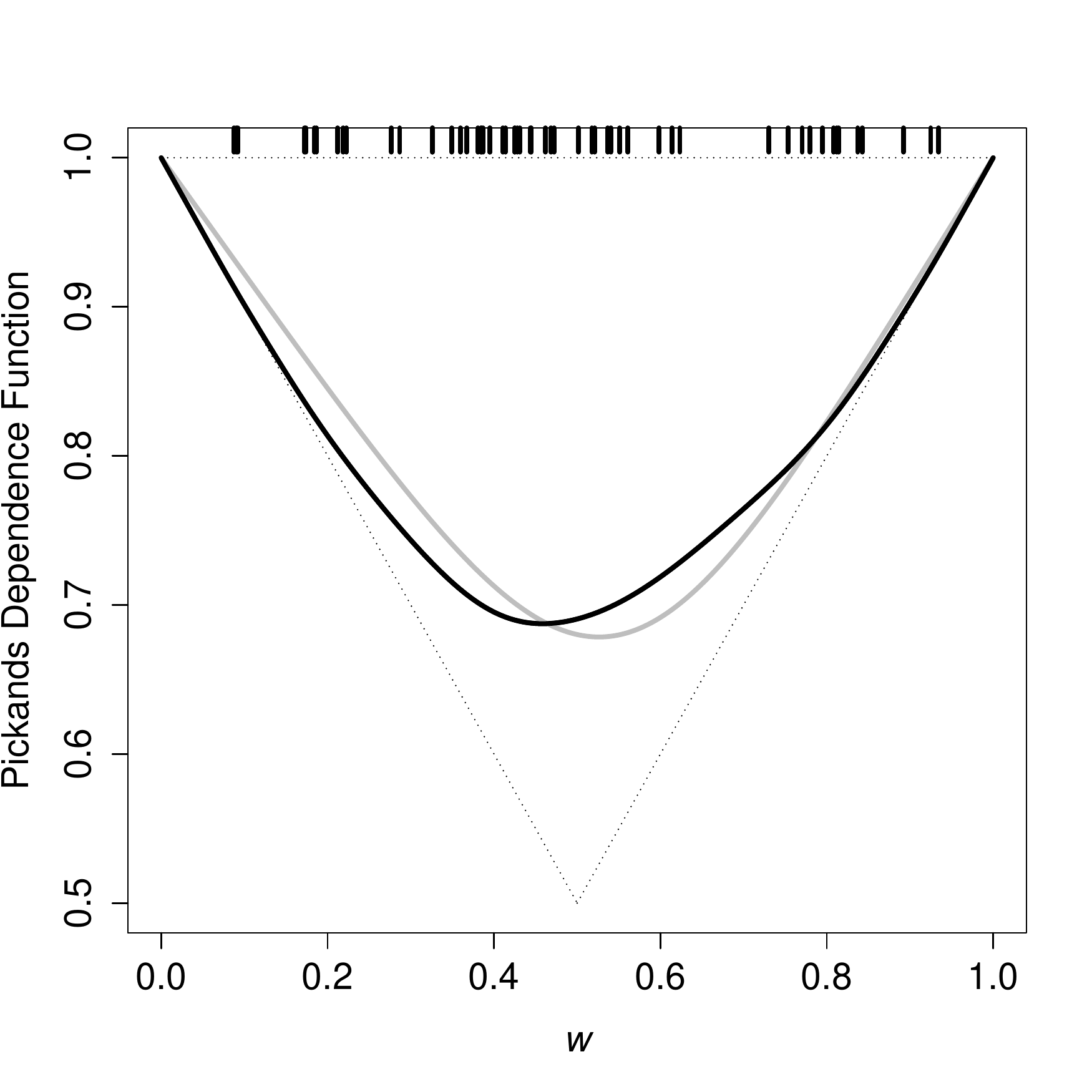}&
\includegraphics[width=.45\textwidth]{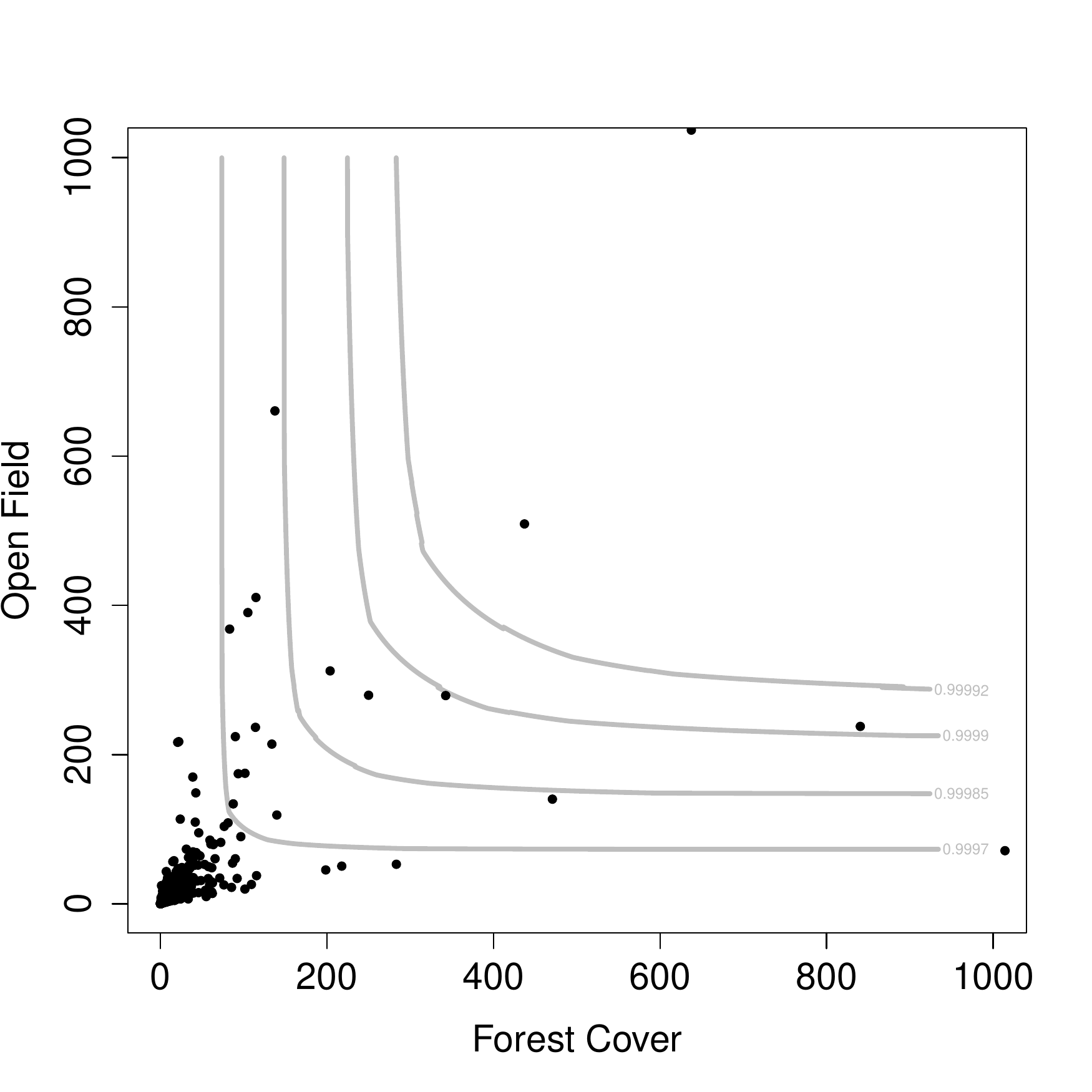}   
\end{tabular}
\caption{\small Estimates of the Pickands dependence function
    and contours of the bivariate extreme value distribution. Left: The
    solid line represents the smooth Pickands dependence function
    obtained from the Euclidean likelihood weights using
    \eqref{A.smooth}, and the gray line represents the fit
    from the asymmetric logistic model with
    $(\widehat{\psi}_1,\widehat{\psi}_2)=(0.78,0.90)$ and
    $\widehat{\rho} =  0.30$. The pseudo-angles are
    presented at the top. Right: The solid gray line represents the smooth bivariate extreme value
    distribution obtained from the Euclidean likelihood weights using
    \eqref{G.smooth}. Air temperature data are plotted on the unit Fr\'echet scale. \label{fig:estimates2}}
\end{center}
\end{figure} 

\section{Discussion}
In this paper we propose a simple empirical likelihood-based estimator for the spectral
measure, whose asymptotic efficiency is comparable to
the empirical likelihood spectral measure of \cite{Einmahl:2009p327}.
The fact that our estimator has the same limit distribution as the 
empirical likelihood spectral measure, suggests that a
more general result may hold for other members of the Cressie--Read
class, of which these estimators are particular cases, similarly to
what was established by \cite{Baggerly:1998} in a context different
from ours. We focus on the spectral measure defined over
the $L_1$-norm, but only for a matter of simplicity, and there is no
problem in defining our estimator for the spectral measure defined
over the $L_p$-norm, with $p \in [1,\infty]$. For real data applications smooth versions of empirical 
the estimator may be preferred, but these can be readily constructed
by suitably convoluting the weights of our empirical likelihood-based method with a kernel
on the simplex. 

\section*{Acknowledgements}
We thank Anthony Davison, Vanda In\'acio, Feridun Turkman, and Jacques Ferrez for
discussions and we thank the editors and anonymous referees for helpful 
suggestions and recommendations, that led to a significant improvement
of an earlier version of this article. Miguel de Carvalho's research
was partially supported by the Swiss National Science
Foundation, CCES project EXTREMES, and by the Funda\c c\~ao para a Ci\^encia e a
Tecnologia (Portuguese Foundation for Science and Technology) through
PEst-OE/MAT/UI0297/2011 (CMA). Johan Segers's research was supported by IAP research network grant No.\ P6/03 of the Belgian government (Belgian Science Policy) and by contract No.\ 07/12/002 of the Projet d'Actions de Recherche Concert\'ees of the Communaut\'e fran\c{c}aise de Belgique, granted by the Acad\'emie universitaire Louvain.

\appendix

\section{Proofs}
\label{app:proofs}

\begin{proof}[Proof of Theorem~\ref{thm:consistency}]
Let $F_n, F \in \mathbb{D}_\Phi$. By Fubini's theorem,
\begin{align*}
  \int_{[0,w]} v \, \mathrm{d}F(v) &= w \, F(w) - \int_0^w F(v) \, \mathrm{d}v, \qquad w \in [0, 1], \\
  \int_{[0,1]} v^2 \, \mathrm{d}F(v) &= 1 - \int_0^1 2v \, F(v) \, \mathrm{d}v,
\end{align*}
and similarly for $F_n$. It follows that $\| F_n - F \|_\infty \to 0$ implies that $\mu_{F_n} \to \mu_F$, $\sigma_{F_n}^2 \to \sigma_F^2$, and $\int_{[0, w]} v \, \mathrm{d} F_n(v) \to \int_{[0, w]} v \, \mathrm{d} F(v)$ uniformly in $w \in [0, 1]$. Hence $\| \Phi(F_n) - \Phi(F) \|_\infty \to 0$. Therefore, the map $\Phi : \mathbb{D}_\Phi \to \ell^\infty([0,1])$ is continuous. The lemma now follows from the fact that $\Phi(H) = H$ together with the continuous mapping theorem, see Theorem 1.9.5 in \cite{vdv:wellner:1996}.
\end{proof}

\begin{proof}[Proof of Theorem~\ref{thm:an}]
Write
\begin{align*}
  \beta_n &= r_n ( \dot{H}_n - H ), & \gamma_n &= r_n ( \hat{H}_n - H ).
\end{align*}
Let $\mathbb{D}_n$ denote the set of functions $f \in \ell^\infty([0, 1])$ such that $H + r_n^{-1} f$ belongs to $\mathbb{D}_\Phi$. Since $H + r_n^{-1} \beta_n = \dot{H}_n$ takes values in $\mathbb{D}_\Phi$, it follows that $\beta_n$ takes values in $\mathbb{D}_n$. Define $g_n : \mathbb{D}_n \to \ell^\infty([0, 1])$ by
\[
  g_n(f) = r_n \{ \Phi( H + r_n^{-1} f ) - H \}.
\]
Observe that
\[
  g_n(\beta_n) = r_n \{ \Phi( \dot{H}_n ) - H \} = \gamma_n.
\]
Further, define the map $g : \mathcal{C}([0, 1]) \to \mathcal{C}([0, 1])$ by
\[
  (g(f))(w) = f(w) - \sigma_H^{-2} \, \int_0^1 f(v) \, \mathrm{d}v \, \int_0^w (1/2 - v) \, \mathrm{d}H(v), \qquad w \in [0, 1].
\]
A straightforward computation shows that if $f_n \in \mathbb{D}_n$ is such that $\| f_n - f \|_\infty \to 0$ for some $f \in \mathcal{C}([0, 1])$, then $\| g_n(f_n) - g(f) \|_\infty \to 0$; note in particular that $f_n(1) = 0$ and thus also $f(1) = 0$. The extended continuous mapping theorem \citep[Theorem~1.11.1]{vdv:wellner:1996} implies that
\[ 
  \gamma_n = g_n(\beta_n) \to g(\beta) = \gamma, \qquad n \to \infty,
\]
as required. Note that we have actually shown that $\Phi$ is Hadamard differentiable at $H$ tangentially to $\mathcal{C}([0, 1])$ with derivative given by $\Phi_F' = g$. The result then also follows from the functional delta method.
\end{proof}

\section{Beta-Kernel Smoothing of Discrete Spectral Measures}
\label{app:smooth}

We only consider the case of the empirical Euclidean spectral
measure using a Beta kernel, but the same applies to the empirical
likelihood spectral measure by replacing the $\widehat{p_i}$ with $\ddot{p}_i$ in \eqref{tilde.p_i}.
The smooth Euclidean spectral density is thus defined as
\begin{equation} \label{h.smooth}
  \widetilde{h}(w) = \sum_{i=1}^k \widehat{p}_i \, \beta\{w; w_i \nu, (1-w_i)\nu\}, \qquad w \in (0,1),
\end{equation}
where $\nu>0$ is the concentration parameter (inverse of the squared bandwidth, to be chosen via cross-validation) and $\beta(w; p, q)$ denotes the Beta density with parameters $p, q > 0$. The corresponding smoothed spectral measure is defined as
\begin{equation} \label{H.smooth}
\widetilde{H}(w) = \int_0^w \widetilde{h}(v) \, \mathrm{d}v = \sum_{i=1}^k
\widehat{p}_i \mathcal{B}\{w; w_i \nu, (1-w_i) \nu\}, \qquad w \in [0,1],
\end{equation}
where $\mathcal{B}(w;p,q)$ is the regularized incomplete beta function, with $p,q >0$.
Since
\begin{equation}\label{check}
  \int_0^1 w \, \widetilde{h}(w) \, \mathrm{d}w = \sum_{i=1}^k \widehat{p}_i
  \bigg\{\frac{\nu w_i}{\nu w_i + \nu(1-w_i)} \bigg\} = \sum_{i=1}^k \widehat{p}_i w_i= 1/2,
\end{equation}
the moment constraint is satisfied. Plug-in estimators for the Pickands
dependence function and the bivariate extreme value distribution
follow directly from
\begin{align}
    \widetilde{A}(w) &= 1-w + 2 \sum_{i=1}^k \widehat{p}_i \int_0^w \mathcal{B}\{u; w_i \nu,
    (1-w_i) \nu\}\mathrm{d}u, \qquad w \in [0,1], \label{A.smooth} \\
    \widetilde{G}(x,y) &= \exp\left\{-\frac{2}{k}\sum_{i=1}^k
    \widehat{p}_i \int_0^1 \max\bigg(\frac{u}{x},\frac{1-u}{y}\bigg)
    \beta\{u; w_i \nu, (1-w_i) \nu\}\mathrm{d}u\right\}, \qquad x,y>0 \label{G.smooth}.
\end{align}


\small
  \begin{thebibliography}{00}\setlength{\itemsep}{0mm}

  \bibitem[{Antoine et~al.(2007)Antoine, Bonnal, and
      Renault}]{Anto:Bonn:Rena:on:2007} Antoine, B., Bonnal, H.,
    Renault, E.~(2007). 
    \newblock On the efficient use of the informational
    content of estimating equations: Implied probabilities and
    Euclidean empirical likelihood. 
    \newblock \textit{J. Econometrics} 138(2):
    461--487.

  \bibitem[{Baggerly(1998)Baggerly}]{Baggerly:1998} 
    Baggerly, K.~A.~(1998). 
    \newblock Empirical likelihood as a goodness-of-fit measure. 
    \newblock \textit{Biometrika} 85(3): 535--547.

  \bibitem[{Ballani and Schlather(2011)}]{ballani:schlather:2011}
    Ballani, F., Schlather, M.~(2011).
    \newblock A construction principle for multivariate extreme value distributions.
    \newblock \textit{Biometrika} 98(3): 633--645.

  \bibitem[{Beirlant et~al.(2004)Beirlant, Goegebeur, Segers, and
      Teugels}]{BGST} Beirlant, J., Goegebeur, Y., Segers, J.,
    Teugels, J.~(2004).  
    \newblock \textit{Statistics of Extremes: Theory and
      Applications}. 
    \newblock  New York: Wiley.

  \bibitem[{Boldi and Davison(2007)}]{Bold:Davi:mixt:2007} Boldi,
    M.-O., Davison, A.~C.~(2007). 
    \newblock  A mixture model for multivariate
    extremes. 
    \newblock \textit{J. R. Statist. Soc.} B 69(2): 217--229.

  \bibitem[{Bouezmarni and Rolin(2003)}]{Bouezmarni:Rolin:mixt:2003} 
    Bouezmarni, T., Rolin, J.-M. (2003).
    \newblock Consistency of the beta kernel density function estimator.
    \newblock \textit{Canadian J. Statist.} 31(1): 89--98.

  \bibitem[{Chen(1997)}]{Chen:1997} 
    Chen, S.~X.~(1997). 
    \newblock Empirical likelihood-based kernel density estimation.
    \newblock \textit{Austral. J. Statist} 39(1): 47--56.

  \bibitem[{Chen(1999)}]{Chen:1999} 
    Chen, S. X. (1999). 
  \newblock Beta kernel estimators for density functions. 
  \newblock \textit{Comput. Statist. Data Anal.} 31(2): 131--145.

  \bibitem[{Coles et~al.(1999)Coles, Heffernan, and Tawn}]{coles:heffernan:tawn:1999}
    Coles, S.~G., Heffernan, J., Tawn, J.~A.~(1999).
    \newblock Dependence measures for extreme value analyses.
    \newblock \textit{Extremes} 2(4): 339--365.


  \bibitem[{Cooley et~al.(2010)Cooley, Davis, and
      Naveau}]{Cool:Davi:Nave:pair:2010} 
    Cooley, D., Davis, R.~A., Naveau, P.~(2010). 
    \newblock  The pairwise beta distribution: A flexible
    parametric multivariate model for extremes.
    \newblock \textit{J. Mult. Anal.} 101(9): 2103--2117.

  \bibitem[{Cr\'{e}pet et~al.(2009)Cr\'{e}pet, Harari-Kermadec, and
      Tressou}]{Crpe:Hara:Tres:usin:2009} 
    Cr\'{e}pet, A., Harari-Kermadec, H., Tressou, J.~(2009). 
    \newblock Using empirical
    likelihood to combine data: Application to food risk
    assessment. 
    \newblock \textit{Biometrics} 65(1): 257--266.

  \bibitem[{Einmahl et~al.(2001)Einmahl, de~Haan, and
      Piterbarg}]{Einm:deH:Pite:nonp:2001} Einmahl, J.~H.~J., de~Haan,
    L., Piterbarg, V.~I.~(2001).  
    \newblock  Nonparametric estimation of the
    spectral measure of an extreme value
    distribution. 
    \newblock \textit{Ann. Statist.} 29(5): 1401--1423.

  \bibitem[{Einmahl and Segers(2009)}]{Einmahl:2009p327} Einmahl,
    J.~H.~J., Segers, J.~(2009). 
    \newblock Maximum empirical likelihood
    estimation of the spectral measure of an extreme-value
    distribution. 
    \newblock \textit{Ann. Statist.} 37(5B): 2953--2989.

  \bibitem[{Ferrez et~al.(2011)Ferrez, Davison, and
      Rebetez}]{Ferrez:Davison:Rebetez:2011} Ferrez, J.~F., Davison,
    A.~C., Rebetez, M.~(2011). 
    \newblock Extreme temperature analysis under
    forest cover compared to an open field.  
    \newblock \textit{Agric. Forest
      Meteo.} 151(7): 992--1001.

  \bibitem[{Guillotte et~al.(2011)Guillotte, Perron, and
      Segers}]{guillotte2011non} Guillotte, S., Perron, F., Segers,
    J.~(2011). 
    \newblock Non-parametric Bayesian inference on bivariate
    extremes. \textit{J. R.  Statist. Soc.} B 73(3): 377--406.

  \bibitem[{Hall and Presnell (1999)}]{Hall:Presnell:1999} Hall, P.,
    Presnell, B.~(1999).  
    \newblock  Density estimation under constraints.
    \newblock \textit{J. Comput. Graph. Statist.} 8(2): 259--277.

  \bibitem[{Kotz and Nadarajah(2000)}]{Kotz:Nada:extr:2000} Kotz, S.,
    Nadarajah, S.~(2000). 
    \newblock \textit{Extreme Value Distributions: Theory
      and Applications}. 
    \newblock London: Imperial College Press.

  \bibitem[{Ledford and Tawn(1996)}]{Ledf:Tawn:stat:1996} Ledford,
    A.~W., Tawn, J.~A.~(1996). 
    \newblock Statistics for near independence in
    multivariate extreme values. 
    \newblock \textit{Biometrika} 83(1): 169--187.

  \bibitem[{Lin and Zhang(2001)}]{Lin:Zhan:bloc:2001} Lin, L., Zhang,
    R.~(2001). Blockwise empirical {E}uclidean likelihood for weakly
    dependent processes. \textit{Statist. Prob. Lett.} 53(2):
    143--152.

  \bibitem[{Owen(1991)}]{Owen:empi:1991} 
    Owen, A.~B.~(1991). 
    \newblock Empirical likelihood for linear models. 
    \newblock \textit{Ann. Statist.} 19(4):
    1725--1747.

  \bibitem[{Owen(2001)}]{owen2001empirical} 
    Owen, A.~B.~(2001). 
    \newblock \textit{Empirical {L}ikelihood}. 
    \newblock Boca Raton: Chapman and Hall.

  \bibitem[{Pickands(1981)}]{Pick:mult:1981} 
    Pickands, J., I.~(1981). 
    \newblock  Multivariate extreme value distributions.
    \newblock \textit{Bulletin Int. Statist. Inst., Proc. 43rd Sess.}  (Buenos
    Aires), pp. 859--878. Voorburg, Netherlands: ISI.

  \bibitem[{Qin and Lawless(1994)}]{Qin:Lawl:empi:1994} 
    \newblock Qin, J., Lawless, J.~(1994). 
    \newblock Empirical likelihood and general estimating
    equations. 
    \newblock \textit{Ann. Statist.} 22(1): 300--325.

  \bibitem[{Renaud et al.(2011)}]{Renaud:2011}
    Renaud, V., Innes, J.~L., Dobbertin, M., Rebetez, M.~(2011).
    \newblock Comparison between open-site and below-canopy climatic
    conditions in Switzerland for different types of forests over 10
    years (1998--2007).
    \newblock \textit{Theor. Appl. Climatol.} 105(1--2): 119--127.

  \bibitem[{Renaud and Rebetez(2009)}]{Renaud:Rebetez:2009}
    Renaud, V., Rebetez, M.~(2009).
    \newblock Comparison between open-site and below-canopy climatic
    conditions in Switzerland during the exceptionally hot summer of 2003.
    \newblock \textit{Agric. Forest Meteo.} 149: 873--880.

  \bibitem[{van der Vaart and Wellner(1996)}]{vdv:wellner:1996}
    Van der Vaart, A.~W., Wellner, J.~A.~(1996).
    \newblock \textit{Weak Convergence and Empirical Processes: With Applications to Statistics}.
    \newblock New York: Springer.

    \bibitem[{Warcho\l(2012)}]{Warchol2012} Warcho\l,
      M.~(2012). \textit{Smoothing Methods for Bivariate
      Extremes}. Joint MSc thesis, Ecole Polytechnique F\'ed\'erale
      de Lausanne and Uniwersytet Jagiello\'nski.

  \bibitem[{Xu(1995)}]{Xu:larg:1995} 
    Xu, L.~(1995). 
    \newblock Large sample properties of the empirical {E}uclidean likelihood estimation for
    semiparametric model. 
    \newblock \textit{Chinese J.  Appl. Prob. Statist.} 10: 344--352.
  \end{thebibliography}
\end{document}